\documentclass[english,12pt]{article}
\usepackage[T1]{fontenc}
\usepackage[latin1]{inputenc}
\usepackage{geometry}
\usepackage{float}
\usepackage{mathrsfs}
\usepackage{amsmath}
\usepackage{amssymb}
\usepackage{graphicx}
\usepackage{setspace}
\usepackage{hyperref}
\usepackage[color=yellow]{todonotes}
\hypersetup{
    colorlinks=true,
    linkcolor=blue,
    filecolor=magenta,      
    urlcolor=cyan,
    citecolor = red,
}

\usepackage{caption}

\makeatletter

\usepackage{needspace}

\floatstyle{ruled}
\newfloat{algorithm}{tbp}{loa}
\providecommand{\algorithmname}{Algorithm}
\floatname{algorithm}{\protect\algorithmname}

\usepackage{algorithm}
 \usepackage{algorithmicx}

\usepackage{amsthm}

\usepackage{mathrsfs}

\usepackage{amsfonts}

\usepackage{epsfig}

\usepackage{booktabs}
\usepackage{bm}

\usepackage{mathrsfs}

\usepackage{subcaption}
\usepackage{enumerate}

\@ifundefined{definecolor}{\@ifundefined{definecolor}
 {\@ifundefined{definecolor}
 {\usepackage{color}}{}
}{}
}{}

\newtheorem{theorem}{Theorem}[section]
\newtheorem{lem}{Lemma}[section]
\newtheorem{rem}{Remark}[section]
\newtheorem{prop}{Proposition}[section]

\newcounter{hypA}
\newenvironment{hypA}{\refstepcounter{hypA}\begin{itemize}
  \item[({\bf A\arabic{hypA}})]}{\end{itemize}}
\newcounter{hypB}

\newcounter{hypD}

\usepackage{babel}\date{}

\usepackage{accents}

\makeatother

\definecolor{green_dark}{rgb}{0.0, 0.5, 0.1}

\begin{document}

\begin{center}

\begin{spacing}{1.5}
{\Large \textbf{Unbiased Estimation using a Class of Diffusion Processes}}
\end{spacing}

\vspace{0.4cm}

BY HAMZA RUZAYQAT$^{1}$, ALEXANDROS BESKOS$^{2}$, DAN CRISAN$^{3}$, AJAY JASRA$^{1}$ \& NIKOLAS KANTAS$^{3}$ 

{\footnotesize $^{1}$Applied Mathematics and Computational Science Program,  Computer, Electrical and Mathematical Sciences and Engineering Division, King Abdullah University of Science and Technology, Thuwal, 23955-6900, KSA.}
{\footnotesize E-Mail:\,} \texttt{\emph{\footnotesize hamza.ruzayqat@kaust.edu.sa, ajay.jasra@kaust.edu.sa}}\\
{\footnotesize $^{2}$Department of Statistical Science, University College London, London, WC1E 6BT, UK.}
{\footnotesize E-Mail:\,} \texttt{\emph{\footnotesize a.beskos@ucl.ac.uk}}\\
{\footnotesize $^{3}$Department of Mathematics, Imperial College London, London, SW7 2AZ, UK.}
{\footnotesize E-Mail:\,} \texttt{\emph{\footnotesize d.crisan@ic.ac.uk, n.kantas@ic.ac.uk}}
\end{center}

\begin{abstract}
We study the problem of unbiased estimation of expectations with respect to (w.r.t.)~$\pi$ a  given, general probability measure on $(\mathbb{R}^d,\mathcal{B}(\mathbb{R}^d))$ that is absolutely continuous with respect to a standard Gaussian measure.  We focus on simulation associated to a particular class of  diffusion processes,  sometimes termed the Schr\"odinger-F\"ollmer Sampler, which is a simulation technique that approximates the law of a particular diffusion bridge process $\{X_t\}_{t\in [0,1]}$ on $\mathbb{R}^d$, $d\in \mathbb{N}_0$. 
This latter process is constructed such that, starting at $X_0=0$, one has  $X_1\sim \pi$. Typically, the drift of the diffusion is intractable and, even if it were not, exact sampling of the associated diffusion is not possible. As a result, \cite{sf_orig,jiao} consider a stochastic Euler-Maruyama scheme that allows the development of  biased estimators for expectations w.r.t.~$\pi$. We show that for this methodology to achieve a mean square error of $\mathcal{O}(\epsilon^2)$, for  arbitrary $\epsilon>0$, the associated cost is $\mathcal{O}(\epsilon^{-5})$. We then introduce an alternative approach that provides unbiased estimates of expectations w.r.t.~$\pi$,  that is, it does not suffer from the time discretization bias or the bias related with the approximation of the drift function. We prove that
to achieve a mean square error of $\mathcal{O}(\epsilon^2)$, the associated cost (which is random) is, with high probability, $\mathcal{O}(\epsilon^{-2}|\log(\epsilon)|^{2+\delta})$, for any $\delta>0$. We implement our method on several examples including Bayesian inverse problems.
\\
\noindent \textbf{Corresponding author}: Hamza Ruzayqat. E-mail:
\href{mailto:hamza.ruzayqat@kaust.edu.sa}{hamza.ruzayqat@kaust.edu.sa} 
\end{abstract}

\section{Introduction}

Let $\pi$ be a probability measure on $(\mathbb{R}^d,\mathcal{B}(\mathbb{R}^d))$, $d\in \mathbb{N}$, with positive Lebesgue density -- denoted also $\pi$ -- assumed to be known up-to a normalizing constant. 
In many applications
in applied mathematics and statistics, it is often of interest to compute expectations of $\pi$-integrable functionals, $\varphi:\mathbb{R}^d\rightarrow\mathbb{R}$, that is compute
$\pi(\varphi):=\int_{\mathbb{R}^d}\varphi(x)\pi(x)dx$, see for instance \cite{robert} and the references therein. There are numerous methodologies for the approximation
of $\pi(\varphi)$ often based upon the simulation of Markov processes with the most prominent example being Markov chain Monte Carlo (MCMC).  In this article we consider approximations
based upon independent samples, each of the latter generated via an `embarrassingly'  parallel approach. Such schemes have become rather popular in the recent literature \cite{ub_mc,ub_grad,jacob2}.


We will consider a stochastic differential equation (SDE) on the time domain $t\in[0,1]$, starting from $X_0=0$ and satisfying a terminal constraint $X_1\sim\pi$. This is an instance of a more general problem of assigning an initial and final distribution to a Markov process, 
which was initially formulated by Schr\"odinger in \cite{schrod} and later developed into a general stochastic bridge construction by Jamison in \cite{jamison_75}, 
whereby an additive drift function is used to ensure the terminal constraint $X_1\sim\pi$ will be satisfied. In our case $\pi$ is absolutely continuous w.r.t.~a standard Gaussian measure, 
so this drift will be added to a $d$-dimensional Brownian motion, $\{W_t\}_{t\in[0,1]}$, whose terminal distribution is known: $W_1\sim\mathcal{N}_d(0,I)$,  denoting the $d$-dimensional Gaussian distribution of mean $0$ and  identity covariance. This gives the following $\mathbb{R}^d$-valued diffusion process:
\begin{equation}
\label{eq:main_diff}
dX_t = b(X_t,t)dt + dW_t, \quad  X_0=0,
\end{equation}
with 
$$b(x,t)=\nabla\log\mathbb{E}_{x,t}[f(W_1)],$$
where for any $(x,t)\in\mathbb{R}^d\times[0,1]$, $\nabla$ denotes the gradient w.r.t  $x$, $\mathbb{E}_{x,t}$ denotes the expectation w.r.t~$\{W_s\}_{s\in[t,1]}$ starting at $W_t=x$ and note that $b(x,1)=0$. We remark that  $f$ corresponds to the analogous of a likelihood function for a standard Gaussian prior, i.e. for  $z\in\mathbb{R}^d$ we have
$
f(z) = \pi(z)/\phi(z),
$
with $\phi(z)$ the standard $d-$dimensional Gaussian density.  Using standard manipulations the expression for $b$ simplifies to 
$$
b(x,t) = \nabla \log \mathbb{E}[{f(x+W_{1-t})}] =\frac{\mathbb{E}_\phi[\nabla f(x+\sqrt{1-t}Z)]}{\mathbb{E}_\phi[f(x+\sqrt{1-t}Z)]} = \frac{1}{\sqrt{1-t}}\frac{\mathbb{E}_\phi[Z f(x+\sqrt{1-t}Z)]}{\mathbb{E}_\phi[f(x+\sqrt{1-t}Z)]},
$$
with $\mathbb{E}_\phi$ denoting expectation w.r.t.~a $d$-dimensional standard Gaussian. 
We refer to \cite{jamison_75,pra} for a generalizations of \eqref{eq:main_diff} and more details on a general existence result obtained by means of Girsanov's theorem and more importantly establishing that $X_1\sim\pi$. 
Based on \cite{jamison_75} there is existence of a weak solution of \eqref{eq:main_diff} in $[0,T]$. 
This requires $f$ to be bounded and $\mathbb{E}_{x,t}[f(W_1)]$ to be twice continuously differentiable in $x$ and once in t, i.e. in $\mathcal{C}^{2,1}\left(\mathbb{R}^d\times[ 0,T)\right)$; see \cite{pra} for more details. To ensure the existence of a strong solution of \eqref{eq:main_diff} the drift $b$ needs to satisfy certain conditions; see \cite{jiao} for details.

%

The formulation in \eqref{eq:main_diff} was proposed in \cite{sf_orig,jiao} as a sampling scheme for $\pi$ and as an alternative to MCMC. The authors used the name Schr\"odinger-F\"ollmer Sampler (SFS) inspired by the original Schr\"odinger problem in \cite{schrod} and its links with the entropy based time reversal of SDEs by F\"ollmer \cite{follmer}. The approach of \cite{sf_orig,jiao} is to discretize the process in time via an Euler-Maruyama scheme and to numerically approximate the drift using standard perfect Monte Carlo estimators. This approach can then be parallelized to produce $M\in\mathbb{N}$ independent samples of $X_1$ to approximate $\pi(\varphi)$. This type of embarrassingly parallel estimators are extremely attractive for their computational savings, versus conventional time averages that often appear in standard iterative (non-parallelisable) MCMC simulation. The possibility of massive parallelization in SFS is also quite competitive against various other recent MCMC schemes from `uncorrected' discretized diffusion schemes such as the unadjusted Langevin method \cite{ula1,ula2}. Indeed, here we do not need to concern ourselves with long-time asymptotic behavior, as the solution of \eqref{eq:main_diff} at time 1 is exactly distributed according to $\pi$. In addition, the stochastic bridge in \eqref{eq:main_diff} and its generalizations have been to solve a certain optimal transport problem (see \cite{pra}).  As a result, different sampling schemes have been proposed recently in \cite{schrod_bridge,bortoli} using iterative proportional fitting within Sequential Monte Carlo and generative modeling respectively. Whilst these schemes are interesting and use variants of \eqref{eq:main_diff}, they are iterative in nature and cannot be parallelised to the extent of SFS.

In this article we make several contributions to the SFS, that we now list.
\begin{enumerate}
\item We show that for the method in \cite{sf_orig,jiao} to obtain estimators of $\pi(\varphi)$ achieving a mean square error (MSE) of $\mathcal{O}(\epsilon^2)$, the associated cost is $\mathcal{O}(\epsilon^{-5})$ for an arbitrary $\epsilon>0$.
\item We construct a \emph{doubly randomized} estimator, based upon the ideas in \cite{diffusions}.  This approach delivers unbiased estimates of finite variance (in contrast to the method in \cite{sf_orig,jiao} that is biased). 

\item We show that the  proposed estimator achieves an MSE of $\mathcal{O}(\epsilon^2)$ with an associated random computing cost that is, with high probability, $\mathcal{O}(\epsilon^{-2}|\log(\epsilon)|^{2+\delta})$, for any $\delta>0$. The term high-probability simply means that to achieve the prescribed MSE with the given cost, this latter cost is achieved with probability $1-\zeta$ for some small $\zeta\in(0,1)$. We note that the expected cost of our method is infinite. \autoref{rem:infinite_cost} explains this in detail. 
\item We apply our new approach for several examples, including Bayesian inverse problems, and numerically verify the above theoretical findings.
\end{enumerate}
The significance of our contributions can be explained as follows. In the context of 1.~we establish that due to the Monte Carlo error in the drift and the bias of the time discretization,  one requires a large computational effort to approximate $\pi(\varphi)$ with high precision. Therefore, whilst the trivially parallel nature of the estimator is intuitively appealing, the associated cost can be prohibitive. In 2.~we then consider a methodology to remove the time discretization bias of the Euler-Maruyama scheme,  based upon the randomization schemes of \cite{mcl,rhee}. As the standard approach in those papers cannot be implemented, due to the fact that the drift must be approximated using Monte Carlo, we show that ideas related to \cite{diffusions} can be adapted 
in the context of SFS to overcome biases due to both the time discretisation and the drift approximation. The overall methodology delivers unbiased estimators of finite variance, using only simulation of standard Gaussian random variables. This latter aspect of the new algorithm is particularly interesting, since the approaches for instance in \cite{ub_mc,ub_grad,jacob2}, also deliver unbiased estimators, but one must resort to complex coupling techniques, whereas we show here that such involved constructs are not always needed. In 3., relying on tools from the analysis of time discretized diffusions, we show that our new method provides a substantial reduction in cost over the original method in~\cite{jiao}. 

This article is structured as follows. In Section \ref{sec:algo} we present the approach in \cite{sf_orig,jiao} and our new unbiased algorithm. In Section \ref{sec:theory} we show that a particular version of the Algorithm provides unbiased estimators of finite variance. Section \ref{sec:numerics} contains our numerical results. Appendix~\ref{app:theory} collects some of the technical results are used in Section \ref{sec:theory}.

\section{Algorithm}\label{sec:algo}

The apparent challenges  with the simulation of  the diffusion process in \eqref{eq:main_diff} are, firstly, that the drift function $b(x,t)$ is typically intractable
and, secondly, even if $b(x,t)$ is available point-wise, exact simulation from \eqref{eq:main_diff} is not possible. 

%

\subsection{Approximate SFS using Euler-Maruyama discretization}

Let $\Delta_l=2^{-l}$, with $l\in\mathbb{N}_0$ given. Then, the approach of \cite{sf_orig,jiao} considers the Euler-Maruyama scheme, for $k\in\{0,1,\dots,\Delta_l^{-1}-1\}$:
\begin{equation}
\label{eq:milstein}
\widetilde{X}_{(k+1)\Delta_l}^{l,N} = \widetilde{X}_{k\Delta_l}^{l,N} + \hat{b}(\widetilde{X}_{k\Delta_l}^{l,N},k\Delta_l)\Delta_l + W_{(k+1)\Delta_l} - W_{k\Delta_l} ,
\end{equation}
where independently of all other random variables we have $(W_{(k+1)\Delta_l} - W_{k\Delta_l})\sim\mathcal{N}_d(0,\Delta_l I)$. In the following,  $N\in\mathbb{N}$ will associated to the accuracy of the Monte Carlo estimator of the drift $b$. The quantity $\hat{b}$ is a numerical approximation of $b$ and is defined as:
\begin{equation}\label{eq:b_standard}
\hat{b}(x,t) = \frac{\frac{1}{N}\sum_{i=1}^N \nabla f(x+\sqrt{1-t}Z^i)}{\frac{1}{N}\sum_{i=1}^N f(x+\sqrt{1-t}Z^i)},
\end{equation}
where for $i\in\{1,\dots,N\}$, $Z^i\stackrel{\textrm{i.i.d.}}{\sim}\mathcal{N}_{d}(0,I)$ are random variables independent of the sequence $(W_{(k+1)\Delta_l} - W_{k\Delta_l})$. It should be noted that the Gaussian random variables
are \emph{updated at each simulation time}. As we will see, this re-simulation is not necessary and in some instances  one can substantially improve
the algorithm if such re-simulation is avoided.  The exact method of \cite{jiao} is given in \autoref{alg:basic_method}.

\begin{algorithm}[!h]
Input: number of i.i.d.~replicates, $M\in\mathbb{N}$; number of samples, $N\in\mathbb{N}$, for the approximation of  the drift function~$b$;  level of discretization, $l\in\mathbb{N}_0$. 
\begin{enumerate}
\item[1.] Repeat for $i\in\{1,2,\ldots, M\}$:

\begin{itemize}
\item [a.] Initialise $\widetilde{X}^{l,N}_{0}(i)=0$.
\item[b.] Repeat for  $k\in\{0,1,\ldots, \Delta_l^{-1}-1\}$:
\begin{itemize}
\item[i.] Sample $Z_k^j(i)\stackrel{\textrm{i.i.d.}}{\sim}\mathcal{N}_d(0,I)$,  $j\in\{1,\dots,N\}$,  and compute:
\begin{align*}
\hat{b}(\widetilde{X}_{k\Delta_l}^{l,N}(i),k\Delta_l) = \frac{\frac{1}{N}\sum_{j=1}^N \nabla f(\widetilde{X}_{k\Delta_l}^{l,N}(i)+\sqrt{1-k\Delta_l}Z_k^j(i))}{\frac{1}{N}\sum_{j=1}^N f(\widetilde{X}_{k\Delta_l}^{l,N}(i)+\sqrt{1-k\Delta_l}Z_k^j(i))}.
\end{align*}
\item[ii.] Generate $(W_{(k+1)\Delta_l}(i) - W_{k\Delta_l}(i))\sim\mathcal{N}_d(0,\Delta_lI)$ and set:
$$
\widetilde{X}_{(k+1)\Delta_l}^{l,N}(i) = \widetilde{X}_{k\Delta_l}^{l,N}(i) + \hat{b}(\widetilde{X}_{k\Delta_l}^{l,N}(i),k\Delta_l)\Delta_l + (W_{(k+1)\Delta_l}(i)- W_{k\Delta_l}(i)).
$$
\end{itemize}
\end{itemize}
\item[2.] Return  $\widetilde{X}_1^{l,N}(1),\dots,\widetilde{X}_1^{l,N}(M)$.
\end{enumerate}
\caption{Biased SFS with Euler-Maruyama and i.i.d. Monte Carlo estimation for $b$}
\label{alg:basic_method}
\end{algorithm}

Using \autoref{alg:basic_method} one can compute Monte Carlo estimators of $\pi(\varphi)$, with $\varphi:\mathbb{R}^d\rightarrow\mathbb{R}$ a $\pi$-integrable function, simply by using the sample average:
\begin{align}
\label{eq:MC_est}
\pi^M(\varphi) := \frac{1}{M}\sum_{i=1}^M\varphi(\widetilde{X}_1^{l,N}(i)).
\end{align}
Now, to study the mean square error of this method, we consider the standard Euler-Maruyama discretization:
\begin{equation}\label{eq:milstein_exact}
\widetilde{X}_{(k+1)\Delta_l}^l = \widetilde{X}_{k\Delta_l}^l + b(\widetilde{X}_{k\Delta_l}^l, k\Delta_l)\Delta_l + W_{(k+1)\Delta_l} - W_{k\Delta_l}.
\end{equation}

To assist our analysis, we make the following assumptions.
 For a vector $x\in\mathbb{R}^d$ (resp.~matrix~$A$) we write the $j^{th}$-element (resp.~$(j,k)^{th}$-element)
as $x_j$ (resp.~$A_{jk}$). Also, $\|\cdot\|_1$ is the $L_1$-norm.
\begin{hypA}\label{ass:2}
\begin{enumerate}
\item[a)] There exist $0<\underline{C}<\overline{C}<+\infty$ such that for any $x\in\mathbb{R}^d$
%
\begin{align*}
\underline{C}\leq f(x)  \leq   \overline{C}, \qquad \|\nabla f(x)\|_{1}   \leq  \overline{C}, \qquad
\|\nabla^2 f(x)\|_{1}   \leq   \overline{C}.
\end{align*}
%
\item[b)] There exists $C<+\infty$ such that for any $(x,y)\in\mathbb{R}^{2d}$
$$
\max\,\big\{\,|f(x)-f(y)|,\,\|\nabla f(x)-\nabla f(y)\|_1,
\|\nabla^2 f(x)-\nabla^2 f(y)\|_1\,\big\}\leq C\|x-y\|_2.
$$
\end{enumerate}
\end{hypA}

Below $\textrm{Lip}(\mathbb{R}^d)$ denotes the collection of measurable functions $\varphi:\mathbb{R}^d\to \mathbb{R}$ so that for a  $C<\infty$,  for all $(x,y)\in\mathbb{R}^{2d}$ we have $|\varphi(x)-\varphi(y)|\leq C\|x-y\|_2$, with $\|\cdot\|_2$ denoting the Euclidean norm. Let $\mathcal{B}_b(\mathbb{R}^d)$ be the collection of measurable and bounded functions $\varphi:\mathbb{R}^d\rightarrow\mathbb{R}$.

We then have the result stated in Proposition \ref{prop:simple_prop} below. 
The associated technical result (\autoref{lem:conv_mc_marg} ) can be found in Appendix \ref{app:theory}.

\begin{prop}\label{prop:simple_prop}
Assume (A\ref{ass:2}). Then for any $\varphi\in\textrm{\emph{Lip}}(\mathbb{R}^d)\cap\mathcal{B}_b(\mathbb{R}^d)\cap \mathcal{C}^4\left(\mathbb{R}^d\right)$ there exists a $C<\infty$ such that for any $(l,N,M)\in\mathbb{N}_0\times\mathbb{N}^2$ we have
\begin{align*}
\mathbb{E}\Big[\big(\pi^M(\varphi)-\pi(\varphi)\big)^{2}\Big] \leq C\big(\tfrac{1}{N} + \tfrac{1}{M} + \Delta_l^2\big).
\end{align*}
\end{prop}

\begin{proof}
Throughout the proof, $C$ is a finite, positive constant that  does not depend on $(l,N,M)$,  with a  value that may change from line-to-line.
We have that 
$$
\pi^M(\varphi)-\pi(\varphi) = \tfrac{1}{M}\sum_{i=1}^M\big\{\varphi(\widetilde{X}_1^{l,N}(i))-\varphi(\widetilde{X}_1^{l}(i))\big\} +
\tfrac{1}{M}\sum_{i=1}^M\big\{\varphi(\widetilde{X}_1^{l}(i))-\pi_l(\varphi)\big\} + \big\{\pi_l(\varphi)-\pi(\varphi)\big\}.
$$
Here, $\widetilde{X}_1^{l}(i)$, $i\in\{1,\dots,M\}$, are i.i.d.~samples obtained  via  recursion \eqref{eq:milstein_exact}, starting at $X_0^{l}=0$, up until time instance $1$. Also,  $\pi_l(\varphi)$
is the expectation of $\varphi$ w.r.t.~the law of $\widetilde{X}^{l}_1$. Via the $C_2$-inequality (
$
\mathbb{E}[|a+b|^2] \leq 2 \mathbb{E}[|a|^2] + 2\mathbb{E}[|b|^2]$, where $a, b$ are random variables of finite second moments), we have the upper-bound
\begin{align}
\mathbb{E}&\Big[\big(\pi^M(\varphi)-\pi(\varphi)\big)^2\Big] \leq  C\,\Bigg(
\mathbb{E}\Big[\Big(\tfrac{1}{M}\sum_{i=1}^M\big\{\varphi(\widetilde{X}_1^{l,N}(i))-\varphi(\widetilde{X}_1^{l}(i))\big\}\Big)^2\Big] \nonumber \\ &\qquad \qquad\qquad\qquad\qquad\qquad+
\mathbb{E}\Big[\Big(\tfrac{1}{M}\sum_{i=1}^M\big\{\varphi(\widetilde{X}_1^{l}(i))-\pi_l(\varphi)\big\}\Big)^2\Big] + \big\{\pi_l(\varphi)-\pi(\varphi)\big\}^2
\Bigg).
\label{eq:simp_prop1}
\end{align}
For the first-term on the R.H.S.~of \eqref{eq:simp_prop1} one can use the conditional Jensen inequality, the Lipschitz property of $\varphi$ followed by \autoref{lem:conv_mc_marg} (note that in the latter result, the fact that the
$Z^1,\dots,Z^N$ are refreshed at each time, does not affect the proof, so the result still holds for the recursion used in \autoref{alg:basic_method}), to obtain
$$
\mathbb{E}\Big[\Big(\tfrac{1}{M}\sum_{i=1}^M\big\{\varphi(\widetilde{X}_1^{l,N}(i))-\varphi(\widetilde{X}_1^{l}(i))\big\}\Big)^2\Big]\leq \frac{C}{N}.
$$
For the second term on the R.H.S.~of \eqref{eq:simp_prop1}, one can use standard results for i.i.d.~variables to yield
$$
\mathbb{E}\Big[\Big(\tfrac{1}{M}\sum_{i=1}^M\big\{\varphi(\widetilde{X}_1^{l}(i))-\pi_l(\varphi)\big\}\Big)^2\Big] \leq \frac{C}{M}.
$$ 
For the third term on the R.H.S.~of \eqref{eq:simp_prop1}, standard weak error results for the Euler discretization of diffusions (\cite[Theorem 14.1.5]{kloeden_platen}) give
$$
\big\{\pi_l(\varphi)-\pi(\varphi)\big\}^2 \leq C\Delta_l^2.
$$
The proof is now  complete.
\end{proof}

As a result of \autoref{prop:simple_prop}, to achieve a mean square error of $\mathcal{O}(\epsilon^{2})$, for some  given $\epsilon>0$, one must choose $N=\mathcal{O}(\epsilon^{-2})$, $M=\mathcal{O}(\epsilon^{-2})$ and $l=\mathcal{O}(|\log(\epsilon)|)$, yielding a cost of~$\mathcal{O}(\epsilon^{-5})$.

\begin{rem}
The smoothness requirement of the test function, $\varphi\in \mathcal{C}^4(\mathbb{R}^d),$ is only used in the final step to get a weak error of order 1 for the Euler approximation. Less smoothness will result in lower order, for instance measurable and bounded Lipschitz derivatives would result to the relevant term in the upper bound of  \autoref{prop:simple_prop}  to be  $C \Delta_l$ instead.
\end{rem}

\begin{rem}
Compared to \cite{pra} we impose the more restrictive assumption (A\ref{ass:2}) as the analysis here and in the subsequent results in Section \ref{sec:theory} is clear this way.
In more details, the assumption on $f$ in terms of the upper and lower bounds is one used often in the importance sampling literature (see for example \cite[Assumption 2.2]{mlpf}) and relates essentially to the ability of the Gaussian to mimic the target $\pi$. This boundedness condition is purely qualitative and can be relaxed at quite considerable complications to the proof. The closer $\pi$ is to a Gaussian, the more likely one can verify this assumption.
Recall the drift coefficient can be written as $b(x,t) = \nabla\log h$, where $h(x,t) = \mathbb{E}[f(x+W_{1-t})]$. The function $h$ is smooth (because of the mollification by the heat kernel), but in general $f$ may not satisfy the upper and lower bounded condition required by (A\ref{ass:2}). Moreover, as $t\to 1$, the smoothness will vanish if $f$ is not smooth. Nevertheless, we can treat the general case by working with a "proxy" of $X$, in other words, we apply the algorithm and the theoretical convergence argument to a process $\widetilde{X}$ that satisfies the equation
$$
d\widetilde{X}_t= \widetilde{b}(\widetilde{X},t)dt + dW_t, \qquad \widetilde{X}_0 = 0,
$$
where 
$$
\widetilde{b}(\widetilde{X},t) = \nabla \log \mathbb{E}\left[\tilde{f}(x+W_{1+\epsilon-t}) \right]
$$
where $\tilde{f}$ is the original $f$ "clipped" from above and below: $\tilde{f} = \max\left\{\alpha,\min\{f,1/\alpha\}\right\}
$ and we choose $\epsilon$ and $\alpha$ sufficiently small to assume that the law of $X$ and the law of $\widetilde{X}$ are sufficiently close or that the boundedness deduced in \autoref{prop:simple_prop} remain the same.
\end{rem}

\subsection{Unbiased Estimation}

\cite{mcl,rhee} present a methodology developed in the setting that $Y_{n}\rightarrow Y_{\infty}$ w.r.t.~$L_2$-norm, for squared integrable random variables $\{Y_{n}\}_{n\ge 1}$, $Y_{\infty}$, and the objective is the unbiased estimation of $\mathbb{E}[Y_{\infty}]$. Extensions of the initial approach developed  in \cite{diffusions} will prove useful in the context of the current work.
Before we continue, we shall denote a consistent Monte Carlo based estimator  of $b(x,t)$ with $N$ samples as $\hat{b}_N(x,t)$. Examples include \eqref{eq:b_standard}
or a convergent MCMC algorithm with target probability $\pi_{x,t}$, as we now explain.
For $(x,t)\in\mathbb{R}^d\times[0,1]$ we have 
$$
b(x,t) = \mathbb{E}_{\pi_{x,t}}\left[\frac{\nabla f(x+\sqrt{1-t}Z)}{f(x+\sqrt{1-t}Z)}\right],
$$
where $\mathbb{E}_{\pi_{x,t}}$ denotes expectation w.r.t.~the probability measure
$$
\pi_{x,t}(dz) \propto f(x+\sqrt{1-t}z)\phi(z)dz.
$$
%
Thus, one can obtain a consistent estimator
$\hat{b}(x,t)$ of $b(x,t)$ using e.g.~MCMC methods. The exact form of the estimator is not specified for now.

We first assume access to two integer valued probability distribution $\mathbb{P}_R$ and  $\mathbb{P}_P$ on $\mathbb{N}_0$ both on $\mathbb{N}_0$. Further let $1\leq N_0<N_1<\cdots$ be a sequence of integers such that
$N_p \rightarrow\infty$, as $p\rightarrow\infty$. Higher values of $N_p$ will mean higher accuracy in estimation of $b$ and at the limit this will lead to a perfect estimator. We will use samples of  $\mathbb{P}_R$ and  $\mathbb{P}_P$ to set the number of discretization levels via $l$ and accuracy of $\hat{b}$ via $N_p$ respectively. The objective is to develop an unbiased estimator of $\pi(\varphi)$. Following ideas in \cite[Algorithm 5]{diffusions}, we will now specify a method that aims to overcome both sources of bias we are confronted with, in a way  that computing costs are reduced. 
We achieve debiasing via the `single term estimator' approach, see \cite{rhee}. The core idea of our method is to work with the random variable: 
\begin{align}
\label{eq:ue}
\widehat{\pi(\varphi)} = \frac{\big(\varphi(X_{1}^{L}[N_P])-\varphi(X_{1}^{L-1}[N_P])\big)-\big(\varphi(X_{1}^{L}[N_{P-1}])-\varphi(X_{1}^{L-1}[N_{P-1}]
\big)}{\mathbb{P}_R(L)\mathbb{P}_P(P)}
\end{align} 
with $L\sim \mathbb{P}_R$, $P\sim \mathbb{P}_{P}$. Also, $X_{1}^{l}[N_p]$, for $l, p\in \mathbb{N}_0$, denotes the approximation of $X_1$ obtained via recursion  \eqref{eq:milstein} for time-step $\Delta_l = 2^{-l}$ and the \emph{same} $N_p$ Gaussian variates for the estimation of the drift $b$ at all locations and time instances where it is needed. Critically, 
the four terms in the nominator of \eqref{eq:ue} are carefully coupled. Also, simple conventions apply in the event that $L=0$ or $P=0$. The detailed approach is described in  \autoref{alg:ub_sf_samp}.
Note that in Step 1b.~we assume that the computation of $\hat{b}$ is \emph{dependent} across levels, at coinciding time points. One way to achieve this is to sample
$N_p$ Gaussians and use an estimator of the type \eqref{eq:b_standard} at both levels with the \emph{same} Gaussians fixed once and for all --	 we believe this point is critical as illustrated in \autoref{fig:incr_var}. The figure shows that the variance of the increments $X_1^L[N]-X_1^{L-1}[N]$ decays much faster when the sample $\{Z_i\}_{i=1}^N$ are fixed. In Step 1b(iii), the term 'concatenated Wiener increment' means that the Wiener increment from time $k\Delta_{l-1}$ to $(k+1)\Delta_{l-1}$ at the coarser level $l-1$, for $l\in \mathbb{N}$, is the sum of the two Wiener increments from time $2k\Delta_l$ to $(2k+1)\Delta_l$ and from time $(2k+1)\Delta_l$ to $(2k+2)\Delta_l$ sampled at the finer level $l$, where $k\in \{0,\cdots,\Delta_{l-1}^{-1}-1\}$.
We note that an alternative to the single term estimator is the independent sum estimator, see \cite{rhee}, that often performs better in  simulations; this latter estimator can be used with little extra difficulty in implementation. 

\begin{figure}[h!]
\centering
\subcaptionbox{Fixed Gaussian Samples}{
\includegraphics[width =0.48\textwidth]{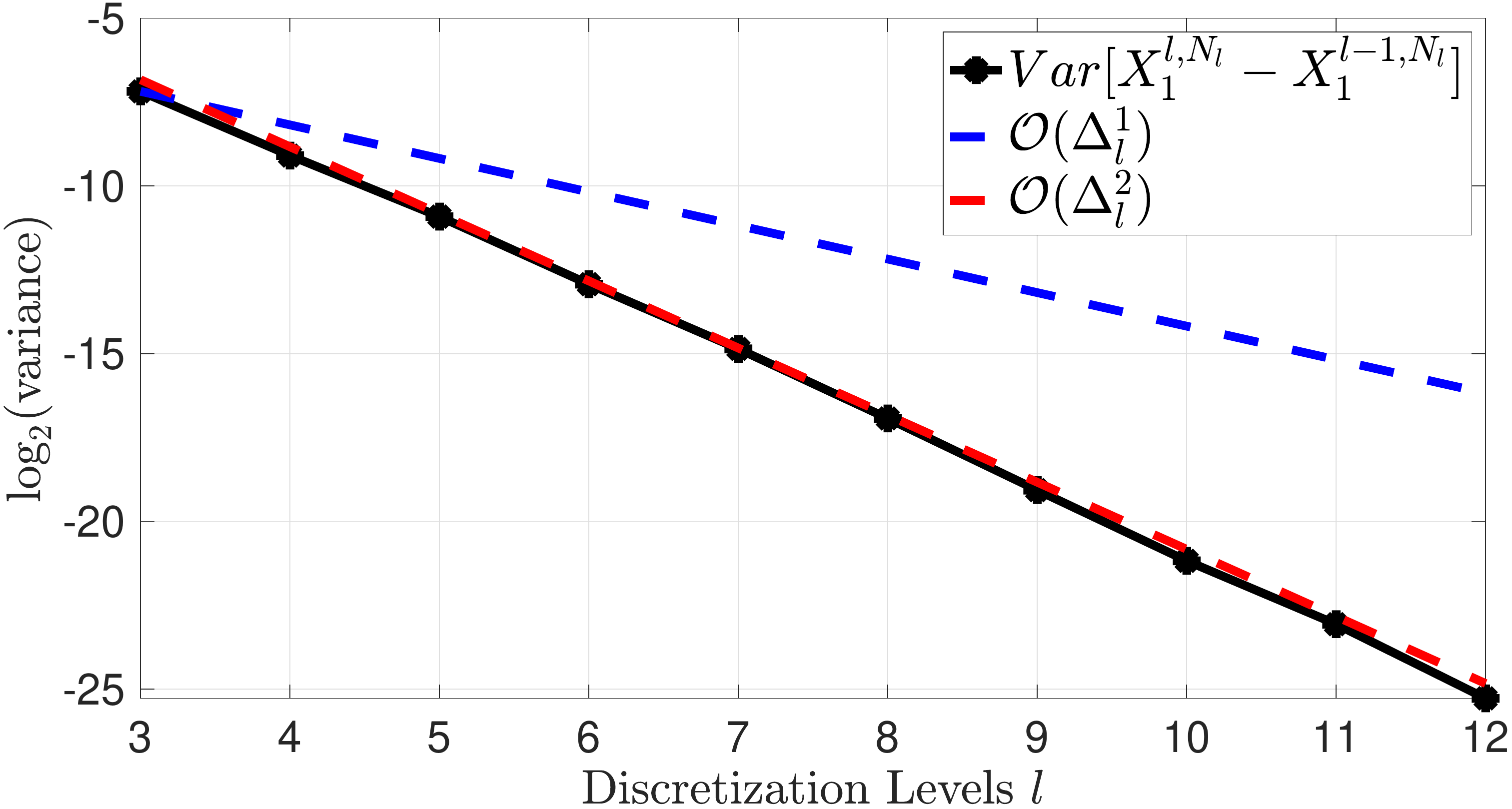}
}
\hfill
\subcaptionbox{Unfixed Gaussian Samples}{
\includegraphics[width =0.48\textwidth]{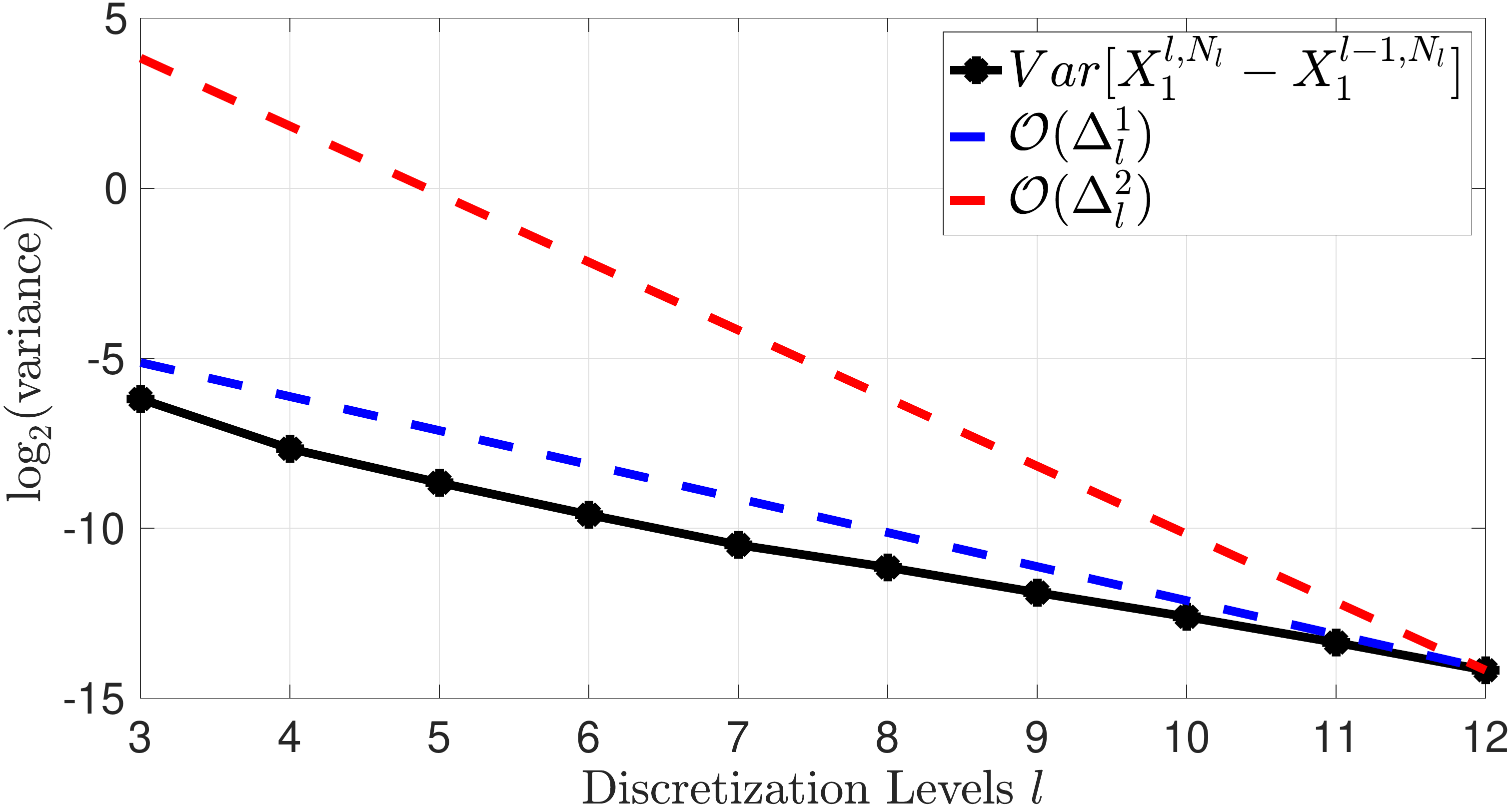}
}
\caption{The log-variance of the level differences estimates for two cases: (a) The samples $\{Z_i\}_{i=1}^N$ are fixed for both levels $l$ \& $l-1$ and all time points. (b) $\{Z_i\}_{i=1}^N$ are sampled at every time point for both levels $l$ \& $l-1$. For simplicity we use a one-dimensional Gaussian density $\pi$.}
\label{fig:incr_var}
\end{figure}

\begin{algorithm}[!h]
Input: number of replicates, $M\in\mathbb{N}$;  sequence $(N_p)_{p\in\mathbb{N}_0}$ and two positive probability mass functions, $\mathbb{P}_R$ and $\mathbb{P}_P$, on $\mathbb{N}_0$.
\begin{enumerate}
\item[1.] Repeat for $i\in \{1,2,\ldots, M\}$:
\begin{itemize}
\item[a.]   Sample $L^i\sim\mathbb{P}_R$ and $P^i\sim \mathbb{P}_P$.
\item[b.]   Sample the following variables.
\begin{itemize}
\item[i.] Sample $N_{P^i}$ Gaussians $\{Z^j(i)\}_{j=1}^{N_{P^i}} \sim \mathcal{N}_d(0,I)$. Sample the Wiener increments $\left\{W_{(k+1)\Delta_{L^i}} - W_{k\Delta_{L^i}}\right\}_{k=0}^{\Delta_{L^i}^{-1}-1}\sim\mathcal{N}_d(0,\Delta_l I)$. Then generate $X_1^{L^i}[N_{P^i}]$ from recursion \eqref{eq:milstein} with $l=L^i$ and $\hat{b}=\hat{b}_{N_{P^i}}$ using the same Gaussian variates, $\{Z^j(i)\}_{j=1}^{N_{P^i}}$, at every time instance.
\item[ii.] Sample $N_{P^i-1}$ Gaussians $\{Z^j(i)\}_{j=1}^{N_{P^i-1}} \sim \mathcal{N}_d(0,I)$. Generate $X_1^{L^i}[N_{P^i-1}]$ from recursion \eqref{eq:milstein} using the same Wiener increments as in (i.) with $l=L^i$ and $\hat{b}=\hat{b}_{N_{P^i-1}}$ using the same Gaussian variates, $\{Z^j(i)\}_{j=1}^{N_{P^i-1}}$, at every time instance. 
\item[iii.] Generate $X_1^{L^i-1}[N_{P^i}]$ from recursion \eqref{eq:milstein} with $l=L^i-1$ and $\hat{b}=\hat{b}_{N_{P^i}}$ -- use the same $N_{P^i}$ Gaussian variates in (i.) and concatenated Wiener increments produced by the ones used in (i.) via the identity:
 \begin{align}
W_{(k+1)\Delta_{L^i-1}}-W_{k\Delta_{L^i-1}} =  (W_{(2k+1)\Delta_{L^i}} - W_{2k\Delta_{L^i}}) + (W_{2(k+1)\Delta_{L^i}} - W_{(2k+1)}\Delta_{L^i}),  \label{eq:conc_inc2}
\end{align}
with $k\in\{0,\dots,\Delta_{L^i-1}^{-1}-1\}$.
\item[iv.] Generate $X_1^{L^i-1}[N_{P^i-1}]$ from recursion \eqref{eq:milstein} with $l=L^i-1$ and $\hat{b}=\hat{b}_{N_{P^i-1}}$ -- use the same $N_{P^i}$ Gaussian variates in (ii.) and the same concatenated Wiener increments as in (iii.). 
\end{itemize}
\item[c.] Set:
\begin{align*}
\widehat{\pi(\varphi)}(i) & =  \frac{\big(\varphi(X_1^{L^i}[N_{P^i}])-\varphi(X_1^{L^i-1}[N_{P^i}])\big) - 
\big(\varphi(X_1^{L^i}[N_{P^i-1}])-\varphi(X_1^{L^i-1}[N_{P^i-1}])\big)}{\mathbb{P}_R(L^i)
\mathbb{P}_P(P^i)}.
\end{align*}
Apply the conventions: \\
If $L^i=0$ then set $\varphi(X_1^{L^i-1}[N_{P^i}]) = \varphi(X_1^{L^i-1}[N_{P^i-1}]) = 0$. \\
If $P^i=0$ then set $\varphi(X_1^{L^i}[N_{P^i-1}]) = \varphi(X_1^{L^i-1}[N_{P^i-1}]) = 0$.
\end{itemize}
\item[2.] Return $\widehat{\pi(\varphi)}(i)$, $i \in \{1,2,\ldots, M\}$.
\end{enumerate}
\caption{Unbiased Estimator of $\pi(\varphi)$.}
\label{alg:ub_sf_samp}
\end{algorithm}

The estimator $\widehat{\pi(\varphi)}$ developed in \autoref{alg:ub_sf_samp} is unbiased. To show that it has finite variance, one strategy is to establish a bound of the type, for fixed $l,p\in \mathbb{N}$:
\begin{equation}\label{eq:bound_needed}
\mathbb{E}\left[
\Big(\varphi(X_1^{l}[N_{p}])-\varphi(X_1^{l-1}[N_{p}])-\varphi(X_1^{l})+\varphi(X_1^{l-1})\Big)^2\right] \leq \frac{C\Delta_l}{N_{p}},
\end{equation}
where $C$ does not depend upon $l,N_p$ and $(X_1^{l},X_1^{l-1})$ are sampled from the exact Euler discretization under the same coupling procedure as the one described in \autoref{alg:ub_sf_samp}. Then, as in \cite{diffusions}, setting  $N_p=2^p$ and $\mathbb{P}_P(l)=\mathbb{P}_R(l)\propto 2^{-l}(l+1)\log_2(l+2)^2$, $M=\epsilon^{-2}$, 
to achieve a variance (the estimator is unbiased) of $\mathcal{O}(\epsilon^2)$, for some $\epsilon>0$ given, the cost with high probability (the cost is random) is  $\mathcal{O}(\epsilon^{-2}|\log(\epsilon)|^{2+\delta})$, for any $\delta>0$. The main challenge from here is to ascertain the bound \eqref{eq:bound_needed}.

\section{Theoretical Results}\label{sec:theory}

\subsection{Verifying the Bound \eqref{eq:bound_needed}}

We now consider a proof of \eqref{eq:bound_needed} for a particular example. To simplify the notations, we will set for $l\in\{0,1,\dots\}$ and  $k\in\{0,1,\dots,\Delta_l^{-1}-1\}$
$$
\widetilde{X}_{(k+1)\Delta_l}^{l,N} = \widetilde{X}_{k\Delta_l}^{l,N} + \hat{b}(\widetilde{X}_{k\Delta_l}^{l,N},k\Delta_l)\Delta_l + W_{(k+1)\Delta_l} - W_{k\Delta_l},
$$ 
where,  for $x\in\mathbb{R}^d$
$$
\hat{b}(x,k\Delta_l) = \frac{\frac{1}{N}\sum_{i=1}^N \nabla f(x+\sqrt{1-k\Delta_l}Z^i)}{\frac{1}{N}\sum_{i=1}^N f(x+\sqrt{1-k\Delta_l}Z^i)},
$$
and, for $i\in\{1,\dots,N\}$,  $Z^i\stackrel{\textrm{i.i.d.}}{\sim}\mathcal{N}_d(0,I)$. The exact Euler scheme is such that for $l\in\{0,1,\dots\}$, $k\in\{0,1,\dots,\Delta_l^{-1}\}$,
$$
\widetilde{X}_{(k+1)\Delta_l}^{l} = \widetilde{X}_{k\Delta_l}^{l} + b(\widetilde{X}_{k\Delta_l}^{l},k\Delta_l)\Delta_l + W_{(k+1)\Delta_l} - W_{k\Delta_l},
$$ 
and note that the Brownian motions are all shared for both the approximated and exact Euler discretizations. Now, let $(l,t)\in\mathbb{N}_0\times[0,1]$ be given and define
$\tau_t^l:=\Delta_l\lfloor\frac{t}{\Delta_l}\rfloor$. 

We have the following result. Associated technical details can be found in Appendix \ref{app:theory}.

\begin{theorem}\label{theo:main_thm}
Assume (A\ref{ass:2}). There exists $C<+\infty$ such that for any $(l,N)\in\mathbb{N}^2$ we have
$$
\mathbb{E}\,\left[\Big\|\big(\widetilde{X}_{1}^{l,N}-\widetilde{X}_{1}^{l-1,N}\big)-\big(\widetilde{X}_{1}^{l}-\widetilde{X}_{1}^{l-1}\big)\Big\|_2^2 \right ]\leq \frac{C\Delta_l}{N}.
$$ 
\end{theorem}

\begin{proof}
We will consider for a  fixed $t\in(0,1]$ the quantity
 $$\mathbb{E}\,\left [\Big\|\big(\widetilde{X}_{\tau_t^l}^{l,N}-\widetilde{X}_{\tau_t^{l-1}}^{l-1,N}\big) -\big(\widetilde{X}_{\tau_t^l}^{l}-\widetilde{X}_{\tau_t^{l-1}}^{l-1}\big)\Big\|^2_2\right ]$$
 and apply a Gr\"onwall inequality type argument. Throughout all of our proofs $C$ is a generic finite constant with value the can change upon each appearance, but will not depend upon $(l,N,t)$.
We have that for $t\in(0,1]$ 
\begin{align}
\label{eq:start}
\mathbb{E}\,\left [\big\| \,\big(\widetilde{X}_{\tau_t^l}^{l,N}-\widetilde{X}_{\tau_t^{l-1}}^{l-1,N}\big) -\big(\widetilde{X}_{\tau_t^l}^{l}-\widetilde{X}_{\tau_t^{l-1}}^{l-1}\big)\,\big\|_2^2\right ] =  \mathbb{E}\,\left [\big\|\,T_1 + T_2\,\big\|_2^2\right ],
\end{align}
where we have defined
\begin{align*}
T_1&:= \int_{\tau_t^{l-1}}^{\tau_t^{l}}\big(\hat{b}(\widetilde{X}_{\tau_s^{l}}^{l,N},\tau_s^{l})-b(\widetilde{X}_{\tau_s^{l}}^{l},\tau_s^{l})\big)ds, \\ 
T_2&:= \int_0^{\tau_t^{l-1}}\Big(\big(\hat{b}(\widetilde{X}_{\tau_s^{l}}^{l,N},\tau_s^{l})-\hat{b}(\widetilde{X}_{\tau_s^{l-1}}^{l-1,N},\tau_s^{l-1})\big)- \big(b(\widetilde{X}_{\tau_s^l}^{l},\tau_s^l)-b(\widetilde{X}_{\tau_s^{l-1}}^{l-1},\tau_s^{l-1})\big)\Big)ds.
\end{align*}
By the $C_2$    and Jensen inequality, we have that
\begin{align*}
\mathbb{E}\,\left [\big\|T_1\big\|_2^2\right ]&\le  C\,(\tau_t^{l}-\tau_t^{l-1})\int_{\tau_t^{l-1}}^{\tau_t^{l}}
\mathbb{E}\,\left [\Big\| \big(\hat{b}(\widetilde{X}_{\tau_s^{l}}^{l,N},\tau_s^{l})-b(\widetilde{X}_{\tau_s^{l}}^{l},\tau_s^{l})\big)\big\|_2^2\right ] ds, \\
\mathbb{E}\,\left [\big\|T_2\big\|_2^2\right ]&\le C\,\mathbb{E}\,\left [\Big\| \int_0^{\tau_t^{l-1}}\Big(\big(\hat{b}(\widetilde{X}_{\tau_s^l}^{l,N},\tau_s^l)-\hat{b}(\widetilde{X}_{\tau_s^{l-1}}^{l-1,N},\tau_s^{l-1})\big) - \big(b(\widetilde{X}_{\tau_s^l}^{l},\tau_s^l)-b(\widetilde{X}_{\tau_s^{l-1}}^{l-1},\tau_s^{l-1})\big)\Big)ds
\big\|_2^2\right ].
\end{align*}
%
We first treat the upper-bound for $\mathbb{E}\, \left [\big\|T_1\big\|_2^2\right ]$. We have that
\begin{align*}
%
&\int_{\tau_t^{l-1}}^{\tau_t^{l}}
\mathbb{E}\, \left [\Big\| \big(\hat{b}(\widetilde{X}_{\tau_s^{l}}^{l,N},\tau_s^{l})-b(\widetilde{X}_{\tau_s^{l}}^{l},\tau_s^{l})\big)\big\|_2^2\right ] ds \le \\
& \qquad C\,\bigg\{\,\int_{\tau_t^{l-1}}^{\tau_t^{l}}\mathbb{E}\, \left [\,\Big\|\big(\hat{b}(\widetilde{X}_{\tau_s^{l}}^{l,N},\tau_s^{l})-
\hat{b}(\widetilde{X}_{\tau_s^{l}}^{l},\tau_s^{l})\big)\big\|_2^2\right ] ds +
\int_{\tau_t^{l-1}}^{\tau_t^{l}}\mathbb{E}\, \left [\Big\|\big(\hat{b}(\widetilde{X}_{\tau_s^{l}}^{l},\tau_s^{l})-b(\widetilde{X}_{\tau_s^{l}}^{l},\tau_s^{l})\big)\big\|_2^2\right ]ds\,
\bigg\}.
\end{align*}
\autoref{lem:stoch_lip} followed by \autoref{lem:conv_mc_marg} give
\begin{align*}
\mathbb{E}\, \left [\Big\|\big(\hat{b}(\widetilde{X}_{\tau_s^{l}}^{l,N},\tau_s^{l})-
\hat{b}(\widetilde{X}_{\tau_s^{l}}^{l},\tau_s^{l})\big)\big\|_2^2\right ] \leq \frac{C}{N},
\end{align*}
with $C$ independent of $s$. Also, application of  \autoref{lem:conv_mc}  gives 
\begin{align*}
\mathbb{E}\, \left [\Big\|\big(\hat{b}(\widetilde{X}_{\tau_s^{l}}^{l},\tau_s^{l})-b(\widetilde{X}_{\tau_s^{l}}^{l},\tau_s^{l})\big)\big\|_2^2\right ] \leq \frac{C}{N}.
\end{align*}
Thus,
\begin{align}\label{eq:nice_bound}
\mathbb{E}\, \left [\big\|T_1\big\|_2^2\right ]
 \leq \frac{C\Delta_l^2}{N}.
\end{align}
We now turn to the upper-bound for $\mathbb{E}\, \left [\big\|T_2\big\|_2^2\right ]$. Adding and subtracting $\hat{b}(\widetilde{X}_{\tau_s^l}^{l},\tau_s^l)-
\hat{b}(\widetilde{X}_{\tau_s^{l-1}}^{l-1},\tau_s^{l-1})$, followed by use of $C_2$ and Jensen inequality, gives
\begin{align*}
%
&\mathbb{E}\, \left [\big\|T_2\big\|_2^2\right ] \leq 
C\,\bigg\{\,
\mathbb{E}\, \left [\Big\|
\int_0^{\tau_{t}^{l-1}}\Big(\big(\hat{b}(\widetilde{X}_{\tau_s^l}^{l,N},\tau_s^l)-\hat{b}(\widetilde{X}_{\tau_s^{l-1}}^{l-1,N},\tau_s^{l-1})\big) - 
\big(\hat{b}(\widetilde{X}_{\tau_s^l}^{l},\tau_s^l)-
\hat{b}(\widetilde{X}_{\tau_s^{l-1}}^{l-1},\tau_s^{l-1})\big)\Big)ds
\big\|_2^2\right ]
\\
&\qquad \quad \quad\quad\quad\quad +\,\mathbb{E}\, \left [\Big\|
\int_0^{\tau_{t}^{l-1}}\Big(
\big(\hat{b}(\widetilde{X}_{\tau_s^l}^{l},\tau_s^l)-\hat{b}(\widetilde{X}_{\tau_s^{l-1}}^{l-1},\tau_s^{l-1})\big) - 
\big(b(\widetilde{X}_{\tau_s^{l}}^{l},\tau_s^{l})-
b(\widetilde{X}_{\tau_s^{l-1}}^{l-1},\tau_s^{l-1})\big)\Big)ds
\big\|_2^2\right ]\,\bigg\}.
\end{align*}
For the first expectation in the above right-hand term one can use \autoref{lem:tech_lem} and for the second expectation one can apply 
\autoref{lem:tech_lem2}, to yield the following upper-bound
\begin{align*}
\mathbb{E}\, \left [\big\|T_2\big\|_2^2\right ] \leq 
C\,\bigg\{\,\frac{\Delta_l}{N} + \int_0^t
\mathbb{E}\, \left [\Big\|\big(\widetilde{X}_{\tau_s^l}^{l,N}-\widetilde{X}_{\tau_s^{l-1}}^{l-1,N}\big) - 
\big(\widetilde{X}_{\tau_s^l}^{l}-\widetilde{X}_{\tau_s^{l-1}}^{l-1}\big)
\big\|_2^2\right ] ds\,\bigg\}.
\end{align*}
The proof is now concluded by applying  Gr\"onwall's inequality and setting $t=1$.
\end{proof}

We will also use the notation that for a differentiable function $\psi:\mathbb{R}^d\rightarrow\mathbb{R}$, $\nabla_k \psi(x)=(\partial \psi/\partial x_k)(x)$, $k\in\{1,\dots,d\}$.

\begin{prop}\label{prop:func_res}
Assume (A\ref{ass:2}). Then for any $\varphi\in \mathcal{B}_b(\mathbb{R}^d)$ with  $(\partial\varphi/\partial x_k)\in\textrm{\emph{Lip}}(\mathbb{R}^d)\cap \mathcal{B}_b(\mathbb{R}^d)$, $k\in\{1,\dots,d\}$ there exists a $C<+\infty$ such that for any $(l,N)\in\mathbb{N}^2$ we have
\begin{align*}
\mathbb{E}\,\Big[\,\Big(\big(\varphi(\widetilde{X}_{1}^{l,N})-\varphi(\widetilde{X}_{1}^{l-1,N})\big)
-\big(\varphi(\widetilde{X}_{1}^{l})-\varphi(\widetilde{X}_{1}^{l-1})\big)\Big)^2\,\Big] \leq \frac{C\Delta_l}{N}.
\end{align*}
\end{prop}

\begin{proof}
For any $(x,y,u,v)\in\mathbb{R}^{4d}$ we have the representation
\begin{align*}
\big(\varphi(x)-\varphi(y)\big)-\big(\varphi(u)-\varphi(v)\big) = T_1 + T_2,
\end{align*}
where we have defined
\begin{align*}
T_1 & =  \sum_{k=1}^d \int_0^1 (\nabla_k \varphi)\big(y+\lambda(x-y)\big)\times \big((x-y)-(u-v)\big)_k \,d\lambda, \\
T_2 & =  \sum_{k=1}^d \int_0^1 \big\{ (\nabla_k \varphi)\big(y+\lambda(x-y)\big)-(\nabla_k  \varphi)\big(v+\lambda(u-v)\big) \big\}\times(u-v)_k \,d\lambda.
\end{align*}
%

Via the $C_2$-inequality, it suffices now to bound the second moments of $T_1$ and $T_2$, when evaluated at $(x,y,u,v)=(\widetilde{X}_{1}^{l,N},\widetilde{X}_{1}^{l-1,N},\widetilde{X}_{1}^{l},\widetilde{X}_{1}^{l-1})$. For $T_1$, since $\nabla_k \varphi$ is bounded, one has the upper-bound
\begin{align*}
\mathbb{E}\,\left [\Big\|\big(\widetilde{X}_{1}^{l,N}-\widetilde{X}_{1}^{l-1,N}\big)-\big(\widetilde{X}_{1}^{l}-\widetilde{X}_{1}^{l-1}\big)\Big\|_2^2\right ]
\end{align*}
that is bounded by $C\Delta_l/N$ via \autoref{theo:main_thm}. For $T_2$, the Lipschitz property of $\nabla_k \varphi$ provides  the upper-bound
\begin{align*}
C\times \mathbb{E}\,\Big[\,\Big(\big\|\widetilde{X}_{1}^{l-1,N}-\widetilde{X}_{1}^{l-1}\big\|_2^2 + \big\|\big(\widetilde{X}_{1}^{l,N}-\widetilde{X}_{1}^{l-1,N}\big)-\big(\widetilde{X}_{1}^{l}-\widetilde{X}_{1}^{l-1}\big)\big\|_2^2\Big)\cdot \big\|\widetilde{X}_{1}^{l}-\widetilde{X}_{1}^{l-1}\big\|_2^2\,\Big].
\end{align*}
For the term
\begin{align*}
\mathbb{E}\,\Big[\,\big\|\widetilde{X}_{1}^{l-1,N}-\widetilde{X}_{1}^{l-1}\big\|_2^2\cdot \big\|\widetilde{X}_{1}^{l}-\widetilde{X}_{1}^{l-1}\big\|_2^2\,\Big]
\end{align*}
one can use Cauchy-Schwarz, \autoref{lem:conv_mc_marg} and the convergence of the Euler approximation, to yield a bound of  $C\Delta_l^2/N$.
For the term 
\begin{align*}
\mathbb{E}\,\Big[\,\big\|\big(\widetilde{X}_{1}^{l,N}-\widetilde{X}_{1}^{l-1,N}\big)-\big(\widetilde{X}_{1}^{l}-\widetilde{X}_{1}^{l-1}\big) \big\|_2^2\cdot \big\|\widetilde{X}_{1}^{l}-\widetilde{X}_{1}^{l-1}\big\|_2^2\,\Big]
\end{align*}
one can use a similar argument to yield a bound of  $C\Delta_l^2/N$. This concludes the proof.
\end{proof}

\begin{rem}
\label{rem:infinite_cost}
We expect that the bound in \autoref{theo:main_thm} (and hence \autoref{prop:func_res}) can be made sharper to $\mathcal{O}(\Delta_l^2/N)$. However, even with this result, one could not obtain finite variance
and finite expected cost, as the rate of convergence of Monte Carlo estimators in $\mathbb{L}_2$ is $\mathcal{O}(N^{-1})$. That is, to ensure that the variance and expected cost are simultaneously finite, one needs any positive probability mass function $\mathbb{P}_{R,P}$ on $\mathbb{N}_0^2$ such that
\begin{align*}
&\sum_{(l,p)\in \mathbb{N}_0^2} \frac{\Delta_l^2}{N_p\mathbb{P}_{R,P}(l,p)}  < \infty, \\[0.2cm]
&\sum_{(l,p)\in \mathbb{N}_0^2} \Delta_l^{-1}N_p\mathbb{P}_{R,P}(l,p)  < \infty.
\end{align*}
There is no $\mathbb{P}_{R,P}$ where the above conditions are satisfied simultaneously. If both inequalities were true, a straightforward application of the Cauchy-Schwartz inequality would give $\sum_{(l,p)\in \mathbb{N}_0^2}\sqrt{\Delta_l}<\infty$, which cannot hold. One possible way to address this could be to increase the rate in $N$ by using Quasi-Monte Carlo estimators, which is something that we leave to future work. It may also be possible to sharpen the rate associated to the estimator as it is, in terms of $N$, using the ideas in \cite{blan}. However, the context of \cite{blan} is much simpler than that considered here and we expect that such a task to be rather arduous. 
\end{rem}

\Needspace{5\baselineskip}
\section{Numerical Results}\label{sec:numerics}

\subsection{MSE-to-Cost Rates}
We now seek to verify \autoref{theo:main_thm} by computing the MSE-to-cost rates. Given $\varphi(x) =x$, we estimate the expectation $\pi(\varphi)$ by first running \autoref{alg:basic_method} with both fixed and unfixed Gaussians to return the estimator in \eqref{eq:MC_est}, which we denote by $\pi^{M,MC}(\varphi)$, second a multilevel estimator defined by the following collapsing sum identity, with Gaussians $\{Z_i\}_{i=1}^N$ are fixed for both levels $l$ \& $l-1$, 
$$
\pi^{M,ML}(\varphi) := \frac{1}{M}\sum_{i=1}^M\sum_{l=L_*}^L  \frac{1}{N_l}  \left\{ \varphi(\widetilde{X}_1^{l,N_l}(i)) -  \varphi(\widetilde{X}_1^{l-1,N_l}(i))\right\},
$$
with the convention that $ \varphi(\widetilde{X}_1^{-1,N_l}(i)) =0$. Here $L_*\in\mathbb{N}_0$ is a starting level of discretization, $L$ is the target level and $N_l = \mathcal{O}((L+1)2^{2L-l})$ (see e.g. \cite{mlpf} for more details on the choice of $N_l$). Finally, we run the unbiased estimator which returns 
$$
\pi^{M,UB}(\varphi) :=\frac{1}{M}\sum_{i=1}^M{\widehat{\pi(\varphi)}(i)},
$$ 
where the sequence $\{\widehat{\pi(\varphi)}(i)\}_{i=1}^M$ is the output of  \autoref{alg:ub_sf_samp}.
The MSE is the average of 100 independent simulations of each algorithm given by
$$
MSE = \frac{1}{100} \sum_{j=1}^{100} \left[ \left[\pi^{M,\boldsymbol{\cdot}}(\varphi)\right]_j - \pi(\varphi) \right]^2,
$$
where $\pi(\varphi)$ is the reference expectation. 

\subsubsection{Simulation Settings}
\label{subsubsec:simuls}
In practice, one has to truncate the values of $P$ and $L$ in \autoref{alg:ub_sf_samp}. In all models below, we set $\mathbb{P}_R(L)=2^{-1.5L} \mathbb{I}_{\{L_*,\cdots,L_{\text{max}}\}}(L)$, where $L_*,L_{\text{max}}\in \mathbb{N}_0$ and $L_*<L_{\text{max}}$. Given $L$ sampled from $\mathbb{P}_R$, we sample $P$ from $\mathbb{P}_{P|L}(P|L)=g(P|L)~ \mathbb{I}_{\{P_*,\cdots,P_{\text{max}}\}}(P)$ and set $N=N_0 2^P$ for some $N_0 \in \mathbb{N}$, where
\begin{align*}
g(P|L) =  \left\{\begin{array}{lcl}
2^{4-P} & \text{if} & P \in \{P_*,\cdots,4 \wedge (L_{\text{max}} - L)\}\\
2^{-P}~ P~[\log_2(P)]^2 &\text{if} & P>4
\end{array} \right.
\end{align*}
This is the same choice as in \cite{diffusions}. We consider four different probability densities described below with $\varphi(x)=x$. For a given MSE $=\epsilon^2 >0$, the cost to compute $\pi^M(\varphi)$ using \autoref{alg:basic_method} is 
$$
\mathcal{C}_{\text{Single}} := NM2^L,
$$ 
where it is assumed that the cost to simulate the SDE in \eqref{eq:milstein} at a discretization level $L$ is $2^L$. As shown in \autoref{prop:simple_prop}, in order to have an MSE of order $\mathcal{O}(\epsilon^2)$, $N$, $M$ and $L$ must be chosen such that $N=\mathcal{O}(\epsilon^{-2})$, $M=\mathcal{O}(\epsilon^{-2})$ and $L=\mathcal{O}(|\log_2(\epsilon)|)$. The cost of the multilevel algorithm is 
$$
\mathcal{C}_{\text{ML}} := M\sum_{l=L_*}^L N_l ~2^l,
$$ 
where $N_l=\mathcal{O}((L-L_*+1)2^{2L-l})$ and $M$ \& $L$ as before. The expected cost of the proposed method in \autoref{alg:ub_sf_samp} is
\begin{align*}
\overline{\mathcal{C}_{\text{Ub}}} := \frac{1}{100}\sum_{j=1}^{100} \mathcal{C}_j, \qquad \mathcal{C}_j = \sum_{i=1}^M \text{cost}^i,
\end{align*}
where
\begin{align*}
\text{cost}^i = \left\{\begin{array}{lcr}
N_0 ~2^{L_*} & \text{if} & L^{i}=L_*, ~ P^{i} = P_*\\
N_0 ~(2^{L^i} + 2^{L^i-1})  & \text{if} & L^i>L_*, ~ P^i = P_*\\
(N_{P^i} + N_{P^i-1} ) ~2^{L_*} &\text{if} & L^i=L_*, ~ P^i > P_*\\
(N_{P^i} + N_{P^i-1} ) ~(2^{L^i} + 2^{L^i-1}) & \text{if} & L^i>L_*, ~ P^i > P_*\\
\end{array} \right.,
\end{align*}
with $L^i\sim \mathbb{P}_R$, and $N_{P^i}=N_0 ~2^{P^i}$, with $P^i\sim \mathbb{P}_{P|L}$. The values of $P_*$, $L_*$, $P_{\text{max}}$, $L_{\text{max}}$ and $N_0$ for the models below are described in \autoref{tab:params}.

\begin{table}[h!]
\centering
\begin{tabular}{|l|l|l|l|l|l|}
\hline
Model & $P_*$  & $P_{\text{max}}$ & $L_*$ & $L_{\text{max}}$ & $N_0$ \\ \hline \hline
One-dimensional Gaussian    &  2     &  15     &     1     &    8          &   10    \\ \hline
Two-dimensional Gaussian mixture      &   1  &  9     &    2    &    8     &   10       \\ \hline
Bayesian logistic regression      &      1     &  9     &     3     &    9    &   10      \\ \hline
Double-well potential  & 1 & 9 & 6 & 12 & 8 \\ \hline
\end{tabular}
\caption{The choices of the different parameters in Algorithm \ref{alg:ub_sf_samp} for each model.}
\label{tab:params}
\end{table}

\subsubsection{Models}

\label{subsubsec: models}

\begin{enumerate}[(a)]
\item \textbf{One-Dimensional Gaussian Distribution:}
\\
In the first example, we take $\pi(x)=\phi(x;1,2)$, a one-dimensional normal Gaussian density with mean 1 (which is the reference expectation) and variance 2. Clearly the true reference is $\pi(\varphi) = 1$.

\item \textbf{Two-Dimensional Gaussian Mixture Distribution:}
\\
We consider a two-dimensional Gaussian mixture distribution with density
given by
$$
\pi(x) =\frac{1}{16} \sum_{i=1}^{16} \phi(x;\mu_i, 0.03~I_2),
$$
where $I_2$ is the $2\times 2$ identity matrix and
$\mu= (\mu_1,\cdots,\mu_{16}) = \{-1,-0.5,0.5,1\} \times \{-1,-0.5,0.5,1\}$ where $\mu_i\in \mathbb{R}^2$. The reference expectation in this example is $\pi(\varphi)=\frac{1}{16}\sum_{i=1}^{16} \mu_i$.

\item \textbf{Bayesian Logistic Regression:}
\\
Next we consider the binary logistic regression in which the binary observations $\{Y_i\}_{i=1}^n$
are conditionally independent Bernoulli random variables such that $Y_i \in \{0,1\}$ and
$$
\mathbb{P}(Y_i=1|X_i = x_i, \beta) = \rho(\beta^T x_i),
$$
where $\rho: \mathbb{R}\to (0,1)$ defined by $\rho(w) = e^w/(1+e^w)$ is the logistic function and $X_i$ and $\beta$ in $\mathbb{R}^d$ are the covariates and the unknown regression coefficients, respectively. The prior density for the parameter $\beta$ is a multivariate normal Gaussian given by
$$
pr(\beta) = \phi(\beta;0,\Sigma_\beta)
$$
where $\Sigma_\beta$ is defined through its inverse $\Sigma_\beta^{-1} = \frac{1}{n} (\sum_{i=1}^n X_i X_i^T)$. The covarites vectors $\{X_i\}_{i=1}^n$ are sampled independently from $\mathcal{U}\{-1,1\}^d$ which are then standardized. The density of the posterior distribution of $\beta$ is given by
$$
\pi\left(\beta|\{X_i=x_i,Y_i=y_i\}_{i=1}^n \right) \propto \exp\left( \sum_{i=1}^n \left[y_i \beta^T x_i-\log(1+\exp(\beta^Tx_i)) \right] -\frac{1}{2} \beta^T \Sigma_\beta^{-1} \beta \right). 
$$
We set $d=5$, $n=100$ and sample the binary observations $\{y_i\}_{i=1}^n  \overset{i.i.d.}{\sim} {\textstyle \mathrm {Bernoulli} (\rho(\beta_*^T x_i))}$, where $\beta_* = \bm{1}_d$, a vector of ones. We take the reference expectation to be the mean of 100 simulations of a random-walk Metropolis-Hastings MCMC with $10^7$ samples and a burn-in of $10^3$.

\item \textbf{Double-Well Potential:\\
Finally, we consider sampling from $\pi(x) = \exp{(-U(x))}$, where $U$ is the double-well potential given by
$$
U(x) = \frac{1}{4} \|x\|_2^4 - \frac{1}{2}\|x\|_2^2, \qquad x\in \mathbb{R}^d.
$$
The double-well potential is one of several quartic potentials of substantial importance in quantum mechanics and quantum field theory \cite{double} for the investigation of different physical phenomena or mathematical features. 
\\
In this example we test the algorithms in a high dimensional setting where we set $d=30$. The true reference expectation is $\pi(\varphi) = 0$.
}
\end{enumerate}

\subsubsection{Results}
In \autoref{fig:cost_mse}, we plot the MSE against the cost obtained by running the original SFS method presented in \autoref{alg:basic_method}, where the Gaussians $Z_k^j$, $j\in \{1,\cdots,N\}$, in step 1b.i., either fixed for all $k\in \{0,1,\cdots,\Delta_l^{-1}-1\}$ or sampled for each $k$, the multilevel method with fixed Gaussians and finally the unbiased method presented in \autoref{alg:ub_sf_samp}. From the plots, we observe that to obtain an MSE of order $\mathcal{O}(\epsilon^2)$, for some $\epsilon>0$, the cost of \autoref{alg:ub_sf_samp} is of order $\mathcal{O}(\epsilon^{-2})$ in all models. The cost order is much higher in the case of \autoref{alg:basic_method} with both fixed and unfixed Gaussians. In both situations the cost is of order $\mathcal{O}(\epsilon^{-(5+\alpha)})$, $\alpha>0$, for models (a), (b) \& (d), and $\mathcal{O}(\epsilon^{-4})$ for model (c). The cost order for the multilevel algorithm with fixed Gaussians is almost the same for models (a) \& (b) $\mathcal{O}(\epsilon^{-3.7})$, but notably smaller for model (c) where it is of order $\mathcal{O}(\epsilon^{-3.13})$ and higher for model (d) where it is of order $\mathcal{O}(\epsilon^{-4})$.  In all examples we note that the unbiased method of \autoref{alg:ub_sf_samp} is more efficient for lower MSE followed by the multilevel implementation. It is worth noting that we are only comparing the theoretical costs obtained from the formulae in \autoref{subsubsec:simuls}, not the machine computational time. The algorithm presented here is embarrassingly parallelizable over $M$, making it much faster on multi-core workstations.
\begin{figure}[h!]
\centering
\subcaptionbox{One-dimensional Gaussian density}{
\includegraphics[width =0.495\textwidth]{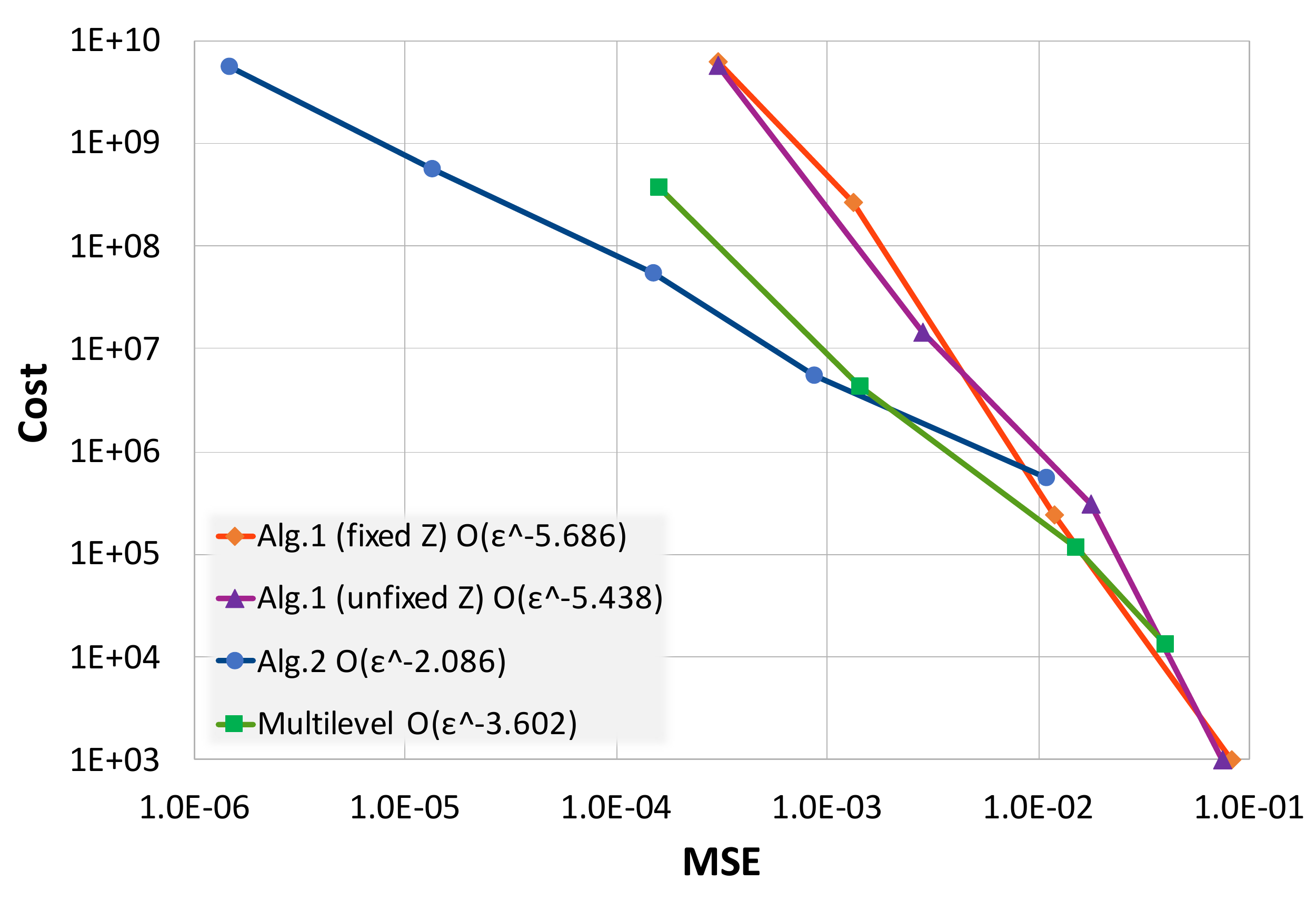}
}
\hspace{-18.5pt}
\subcaptionbox{Two-dimensional mixture of Gaussian densities}{
\includegraphics[width =0.495\textwidth]{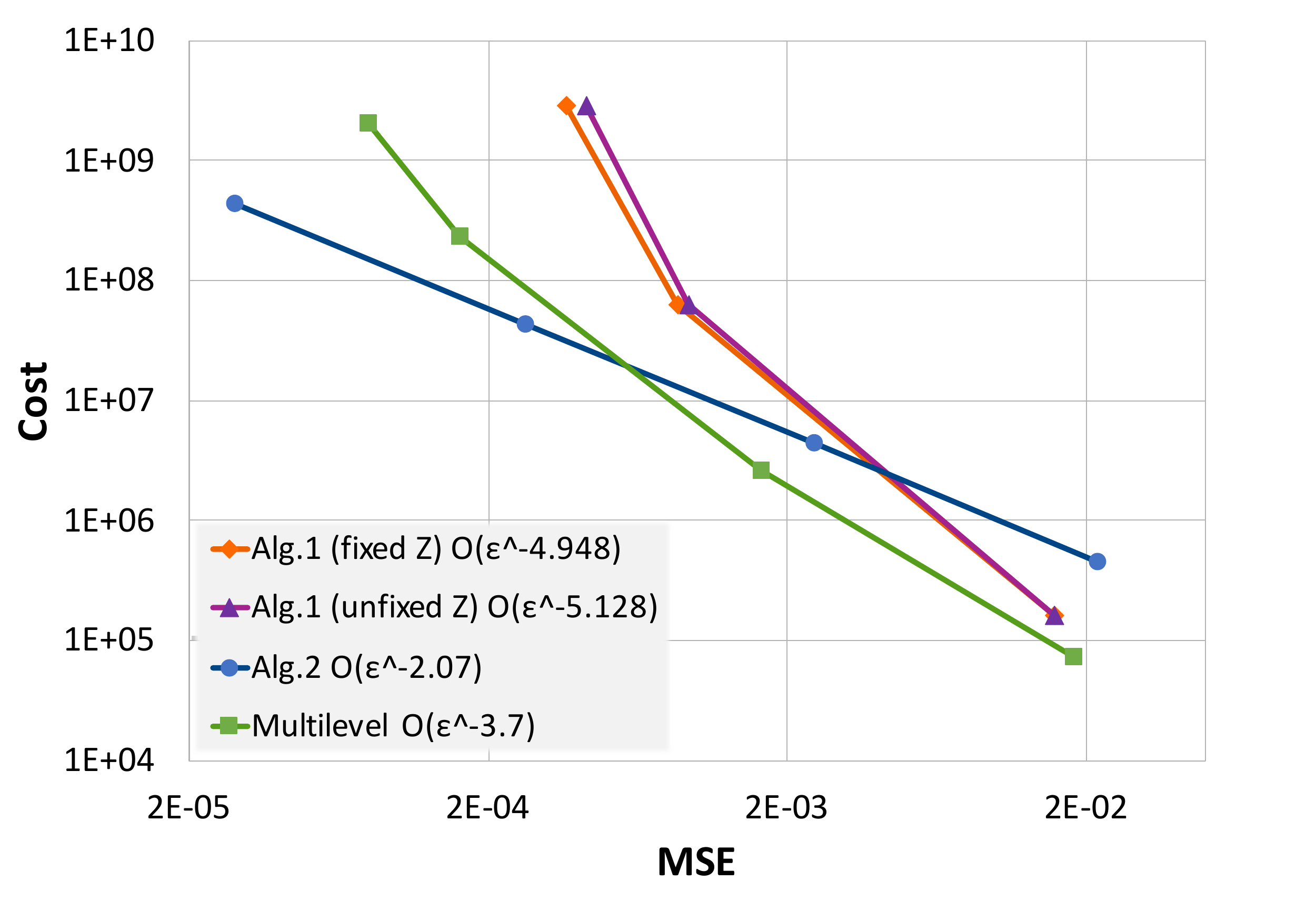}
}\\
\subcaptionbox{Bayesian logistic regression}{
\includegraphics[width =0.48\textwidth]{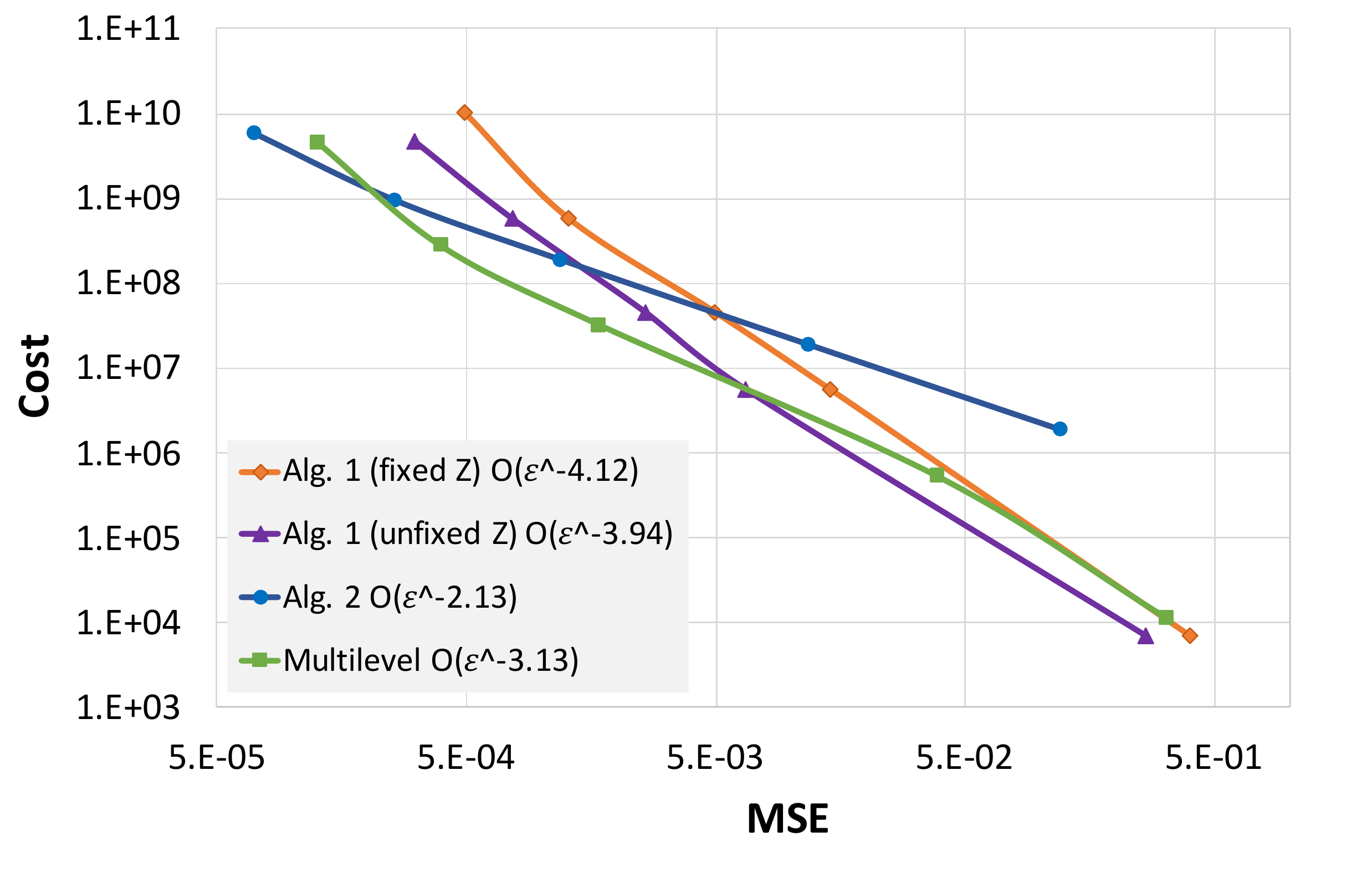}
}
\subcaptionbox{Double-well potential}{
\includegraphics[width =0.48\textwidth]{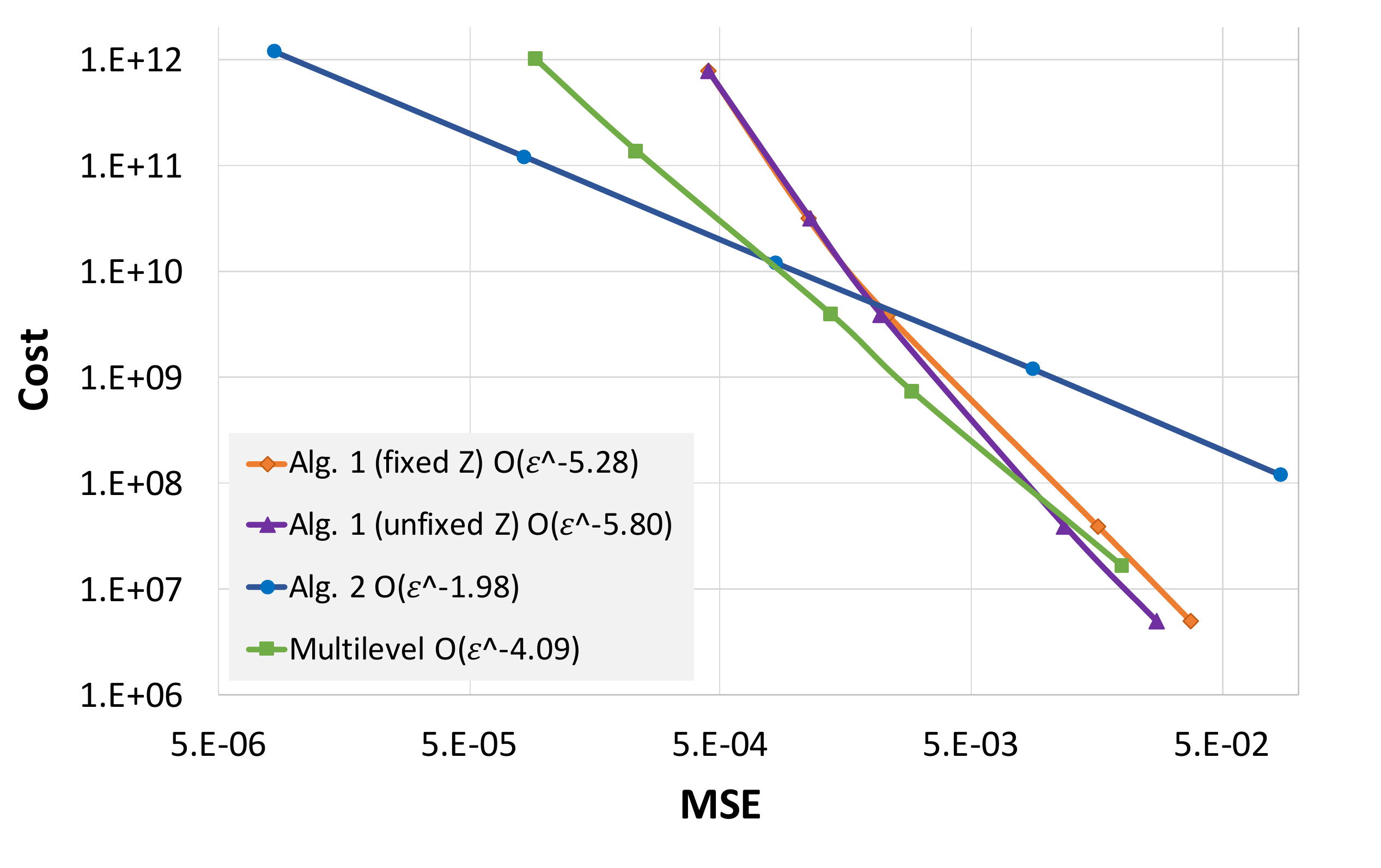}
}
\caption{MSE versus cost (as computed in \autoref{subsubsec:simuls}) of running \autoref{alg:basic_method} with both fixed (orange line) and unfixed (purple line) Gaussians, a multilevel algorithm (green line) and the unbiased estimation \autoref{alg:ub_sf_samp} (blue line).}
\label{fig:cost_mse}
\end{figure}

\subsection{Bayesian Elliptic Inverse Problem}
\label{sec:inv_prob}
The objective of this section is to compare between the original algorithm and the proposed here in the context of Bayesian inverse problems. We consider estimates of expectations w.r.t. the posterior measure on some unknown field of interest using a Bayesian statistics approach of inverse problems that arise from the confluence of partial differential equations and observational data. However, due to floating-point precision limitations and high variance associated with computing $f$, the unbiased method as presented in \autoref{alg:ub_sf_samp} will not work well in general.
Typically, in some cases, especially for high-dimensional distributions with densities that can be expressed as $\pi(x) \propto \exp{(-U(x))}$, the variable $f$, which can be written as $f(x)\propto \exp{(-U(x) + \frac{1}{2}\|x\|^2)}$, will register zero values in fixed point arithmetic, resulting in an infinite drift. Even with the common log-sum-exp trick, the values of $f(x+\sqrt{1-t}Z)$, $Z\sim \phi(z)$, may not be within the representable range of the computing machine. As a result, we advise using an alternative approach provided in Algorithms \ref{alg:new_idea_coupling}--\ref{alg:new_idea_ub}. We should highlight however that this algorithm will not be mathematically investigated in this study; instead, this will be the subject of future research.

\subsubsection{Model}
Consider a Bayesian inverse problem involving inference of the log-permeability coefficient of a 2D elliptic PDE in a bounded and open region $\Omega \subset \mathbb{R}^2$ with convex boundary $\partial\Omega \in C^0$, given noisy measurements of (some components of, or functions of) the associated solution field. In particular, we consider the elliptic PDE
\begin{align}
-\nabla \cdot (K(u) \nabla p) &= F \quad \text{on }\Omega\nonumber\\
 p &= 0 \quad \text{on }\partial\Omega.
 \label{eq:PDE}
\end{align}
where $p$ is the forward state, e.g. the pressure field in a porous media flow, $K(u)$ is a scalar function that represents the permeability filed e.g. of a subsurface rock, and $F$ represents the external force field. This represents a simplified model in groundwater flow. It is important to assume that $K(u)$ is positive in order to have a well-posed problem \cite{stuart}, and hence we write $K(u) = \exp(u)$ and consider the problem of determining $u$ given a set of noisy observations of the pressure in the interior of $\Omega$. Let $(X, \|\cdot \|_X)$ and $(Y,\|\cdot \|_Y)$ be Banach spaces. Then the inverse problem is to determine $u \in X$, given the data
\begin{align}
\label{eq:IP_data}
y = \mathcal{G}(u) + \eta \quad \in Y,
\end{align}
where $ \mathcal{G}: X\to Y$ is the observation operator and $\eta \in \mathcal{N}(0,\Gamma)$ for some trace class, positive, self-adjoint operator $\Gamma$ on $Y$. The Bayesian approach to this inverse problem is to find a posterior probability measure $\mu^y$ on $X$ such that Bayes' rule holds
$$
\frac{d \mu^y}{d \mu_0}(u) \propto l(u;y),
$$
where $l(u;y)$ is the measurement likelihood and $\mu_0$ is the prior, which here is assumed to be Gaussian of the form $\mu_0 = \mathcal{N}(0,\mathcal{C})$, with $\mathcal{C}$ a trace class, positive, self-adjoint operator on $X$. The likelihood can be obtained from \eqref{eq:IP_data} as
$$
l(u;y) = \exp\left(- \frac{1}{2} \|\Gamma^{-1/2}(y-\mathcal{G}(u)) \|^2_Y \right).
$$
Note that if $\mathcal{G}$ is linear and the prior is Gaussian, then the posterior $\mu^y$ will be Gaussian as well, which is the case here as the differential operator associated with the above PDE is linear. 

\begin{algorithm}[!h]
\begin{enumerate}
\item Input: $N\in\mathbb{N}$ the number of samples used to approximate $b$ and a level $l\in\mathbb{N}$ of discretization.
\item For $k=0,1,\cdots,\Delta_l^{-1}-1$, generate the increments of Brownian motion, $W_{(k+1)\Delta_l}^l - W_{k\Delta_l}^l$, used for the Euler approximation at level $l$. Concatenate the increments to generate $W_{(k+1)\Delta_{l-1}}^{l-1} - W_{k\Delta_{l-1}}^{l-1}$ for $k=0, 1,\cdots, \Delta_{l-1}^{-1}-1$. Set $X_0^l=X_0^{l-1}=0$.
\item For $k\in\{0,1,\dots,\Delta_{l-1}^{-1}-1\}$ perform the following:
\begin{itemize}
\item Generate $N$ samples from $\pi_{X_{k\Delta_{l-1}}^{l-1},k\Delta_{l-1}}$ and compute
$$
\hat{b}(X_{k\Delta_{l-1}}^{l-1},k\Delta_{l-1}) = \frac{1}{\sqrt{1-k\Delta_{l-1}}}\frac{1}{N}\sum_{i=1}^N Z^i.
$$

\item Compute:
$$
X_{(k+1)\Delta_{l-1}}^{l-1} = X_{k\Delta_{l-1}}^{l-1}  + \hat{b}(X_{k\Delta_{l-1}}^{l-1} ,k\Delta_{l-1})\Delta_{l-1}+ W_{(k+1)\Delta_{l-1}}^{l-1} - W_{(k+1)\Delta_{l-1}}^{l-1}.
$$

\item Then perform the following for $m=0,1$:
\begin{itemize}
\item Compute the weights and normalize
\begin{align}
w_{[2(k-1)+m]\Delta_l}^i \propto \frac{f \left( X_{[2(k-1)+m]\Delta_l}^l+\sqrt{1-[2(k-1)+m]\Delta_l}Z^i \right)}{
f \left( X_{(k+1)\Delta_{l-1}}^{l-1}+\sqrt{1-[2(k-1)+m]\Delta_l}Z^i \right) }
\label{eq:new_idea_weights}
\end{align}
then compute
\begin{align}
\hat{b}\left(X_{[2(k-1)+m]\Delta_l}^l,[2(k-1)+m]\Delta_l \right) = \sum_{i=1}^N w_{[2(k-1)+m]\Delta_l}^i ~ Z^i. 
\label{eq:new_idea_fine_level}
\end{align}

\item Compute:
\begin{align*}
X_{[2(k-1)+m+1]\Delta_l}^l = & X_{[2(k-1)+m]\Delta_l}^l+ \hat{b}\left(X_{[2(k-1)+m]\Delta_l}^l,[2(k-1)+m]\Delta_l \right)\\
&\hspace{2cm} + W_{[2(k-1)+m+1]\Delta_l}^l - W_{[2(k-1)+m]\Delta_l}^l.
\end{align*}
where the increments of the Brownian motion are concatenated from 2..
\end{itemize}
\end{itemize}
\item Return $X_1^l[N]$ and $X_1^{l-1}[N]$.
\end{enumerate}
\caption{Alternative Coupling for SFS}
\label{alg:new_idea_coupling}
\end{algorithm}

\begin{algorithm}[!h]
Input: number of replicates, $M\in\mathbb{N}$; sequence $(N_p)_{p\in\mathbb{N}_0}$ and two positive probability mass functions, $\mathbb{P}_R$ and $\mathbb{P}_P$, on $\mathbb{N}_0$, $\mathbb{P}_R$ and $\mathbb{P}_P$.
\begin{enumerate}
\item[1.] Repeat for $i\in \{1,2,\ldots, M\}$:
\begin{itemize}
\item[a.] Sample $L^i\sim\mathbb{P}_R$ and $P^i\sim \mathbb{P}_P$.
\item[b.] If $L^i = 0$, generate $X_1^{L^i}[N_{P^i}]$ and $X_1^{L^i}[N_{P^i-1}]$ from recursion \eqref{eq:milstein} using the same Wiener increments with $l=0$ and $\hat{b}$ as in \eqref{eq:new_idea_drift1}.
\item[c.] Otherwise: 
\begin{itemize}
\item[i.] Run \autoref{alg:new_idea_coupling} to return $X_1^{L^i}[N_{P^i}]$ and $X_1^{L^i-1}[N_{P^i}]$.
\item[ii.] Run \autoref{alg:new_idea_coupling} to return $X_1^{L^i}[N_{P^i-1}]$ and $X_1^{L^i-1}[N_{P^i-1}]$ using the same Wiener increments as in (i).

\end{itemize}
\item[c.] Set:
\begin{align*}
\widehat{\pi(\varphi)}(i) & =  \frac{\big(\varphi(X_1^{L^i}[N_{P^i}])-\varphi(X_1^{L^i-1}[N_{P^i}])\big) - 
\big(\varphi(X_1^{L^i}[N_{P^i-1}])-\varphi(X_1^{L^i-1}[N_{P^i-1}])\big)}{\mathbb{P}_R(L^i)
\mathbb{P}_P(P^i)}.
\end{align*}
Apply the conventions: \\
If $L^i=0$ then set $\varphi(X_1^{L^i-1}[N_{P^i}]) = \varphi(X_1^{L^i-1}[N_{P^i-1}]) = 0$. \\
If $P^i=0$ then set $\varphi(X_1^{L^i}[N_{P^i-1}]) = \varphi(X_1^{L^i-1}[N_{P^i-1}]) = 0$.
\end{itemize}
\item[2.] Return $\widehat{\pi(\varphi)}(i)$, $i \in \{1,2,\ldots, M\}$.
\end{enumerate}
\caption{Alternative Unbiased Estimator of $\pi(\varphi)$.}
\label{alg:new_idea_ub}
\end{algorithm}

Clearly, for computer implementation, one needs to discretize the prior, the likelihood, and hence the posterior. A standard finite element method (FEM) with linear triangular elements is employed to solve the forward problem in \eqref{eq:PDE}. The induced mesh consists of right triangles in the domain $\Omega = [0,1] \times [0,1]$. We denote by $\widehat{N}$ the set of nodes (vertices) in the mesh and $\hat{n}$ the number of nodes.

We assume the following additive noise-corrupted pointwise observation model
$$
y_j = p(x_j) + \eta_j, \qquad j = 1,\cdots,J,
$$
where $J$ is the total number of observation locations, $\{x_j\}_{j=1}^J$ the set of nodes of at which the pressure $p$ is observed, $\eta = (\eta_1,\cdots,\eta_J)$ a Gaussian noise distributed according to $\mathcal{N}(0,\Gamma)$, and $y_j$ the actual noise-corrupted observation at node 
$x_j \in \widehat{N}$. For simplicity $\Gamma$ is assumed to equal $\sigma_y^2\, I_J$. A finite-dimensional approximation of the prior is given by $\widehat{\mu}_0=\mathcal{N}(0,C)$, with the entries of the covariance matrix are given by
$$
C_{ij}(x_i,x_j) = \sigma \exp{(-\|x_i-x_j\|_2/\alpha)},
$$
where $x_i,x_j \in \widehat{N}$, $i,j = 1, \cdots,\hat{n}$, $\|\cdot\|_2$ is the Euclidean distance and $\sigma,\alpha >0$ are hyperparameters.

\subsubsection{Alternative Approach to Algorithm \ref{alg:ub_sf_samp}}
\captionsetup[sub]{font=small}
\begin{figure}[h!]
\centering
\subcaptionbox{The reference log-permeability field in 3D}{
\includegraphics[width =0.31\textwidth]{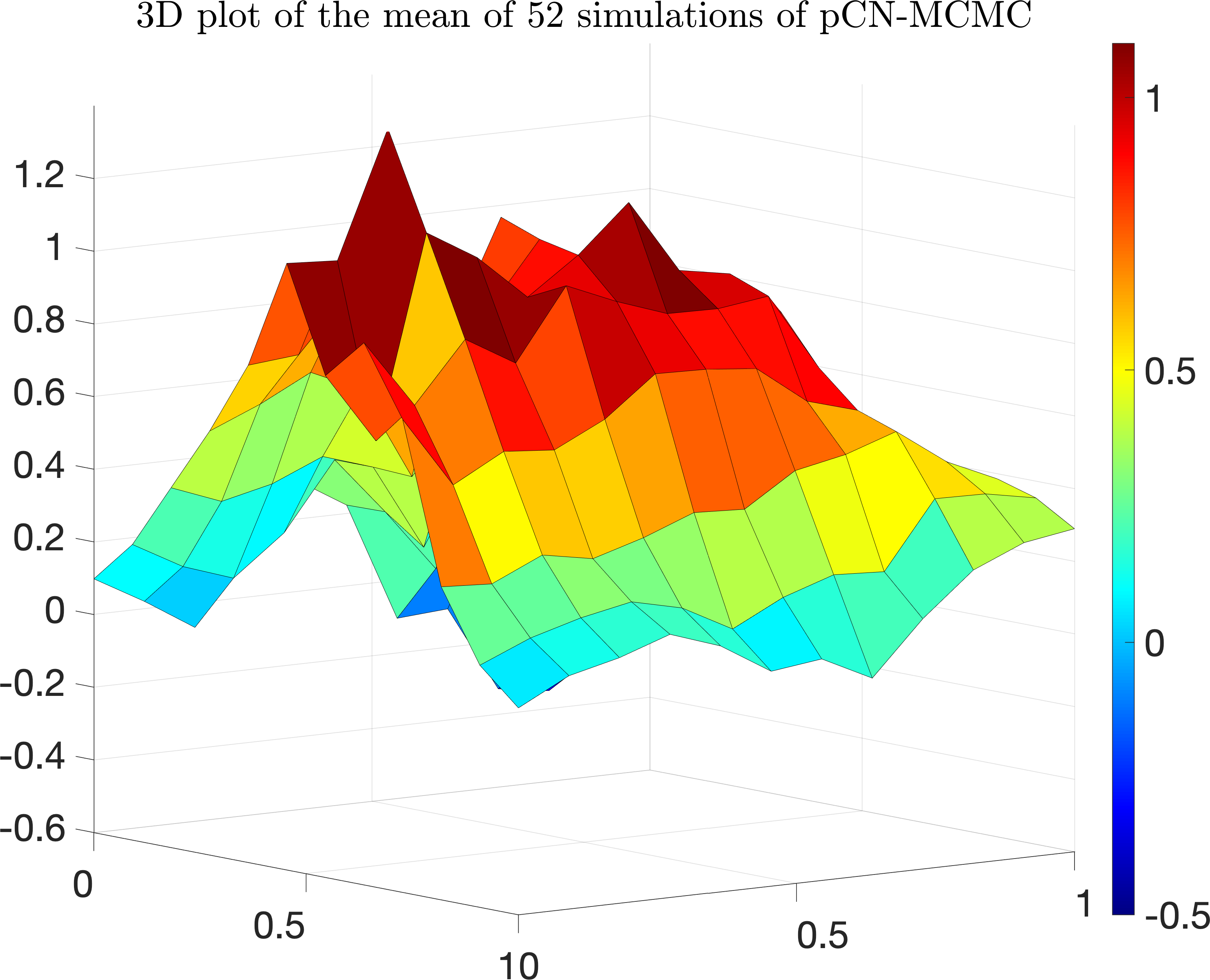}
}
\subcaptionbox{Single-level SFS results in 3D}{
\includegraphics[width =0.31\textwidth]{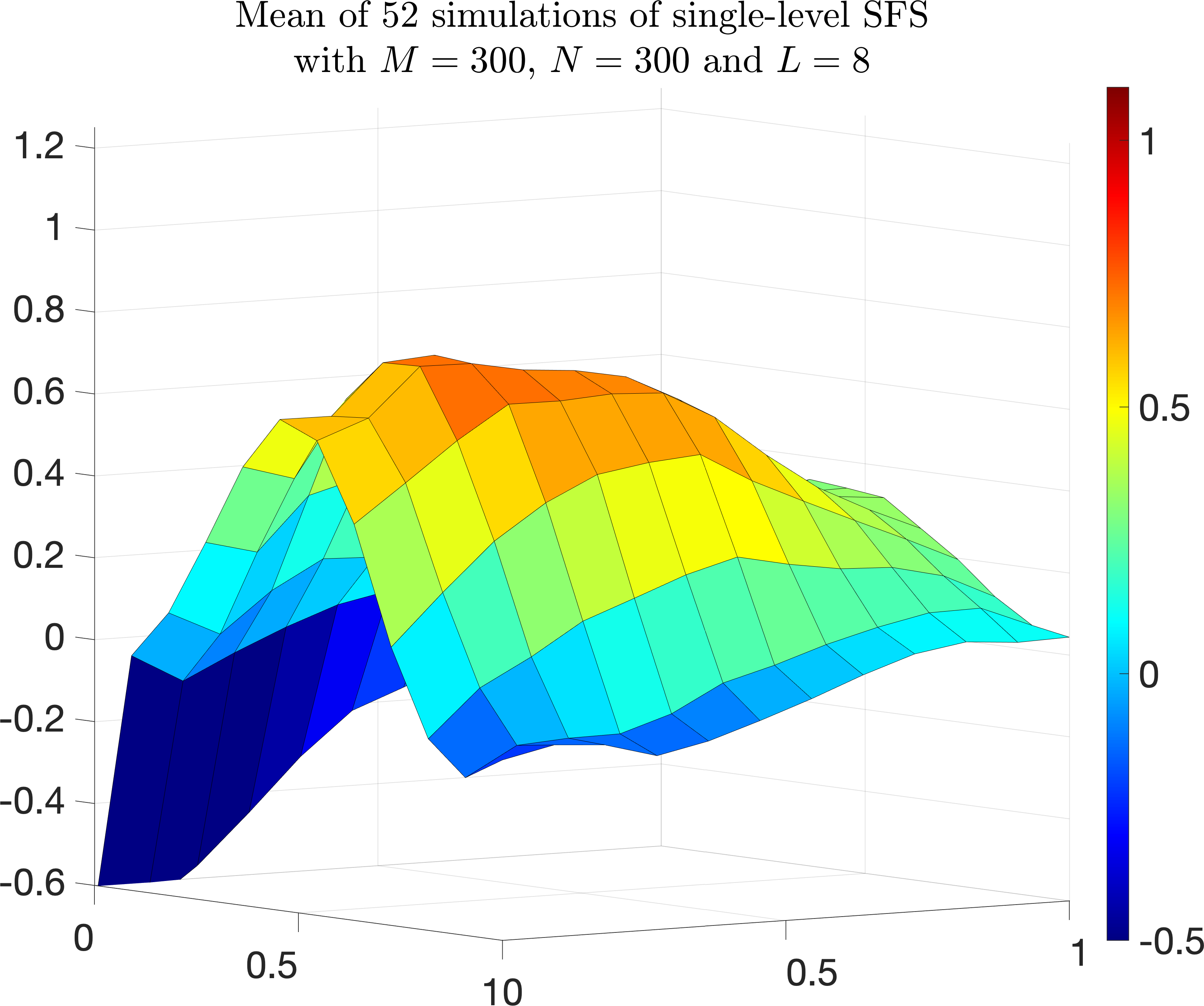}
}
\subcaptionbox{Unbiased SFS results in 3D}{
\includegraphics[width =0.31\textwidth]{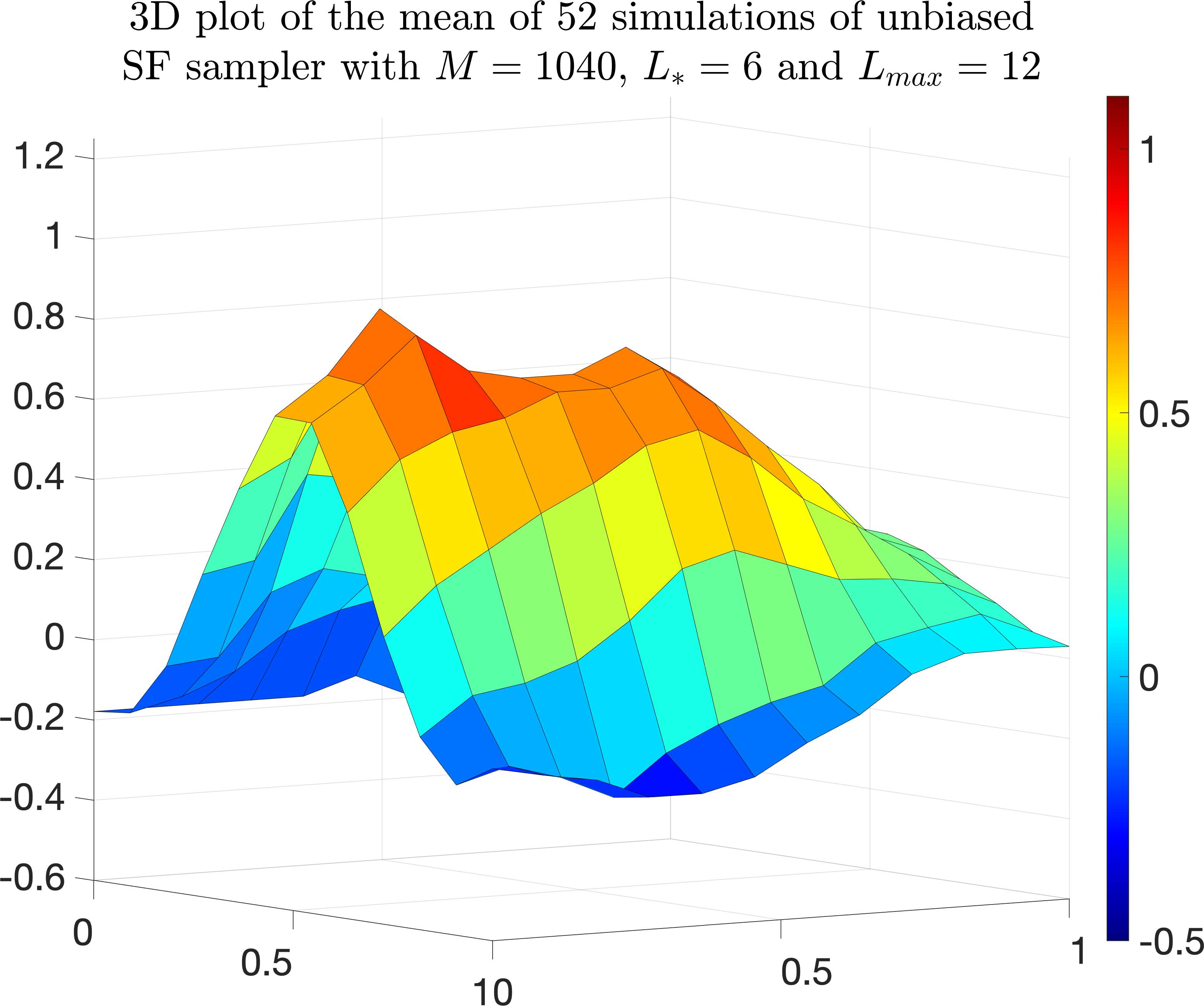}
}

\subcaptionbox{The reference log-permeability field in 2D}{
\includegraphics[width =0.31\textwidth]{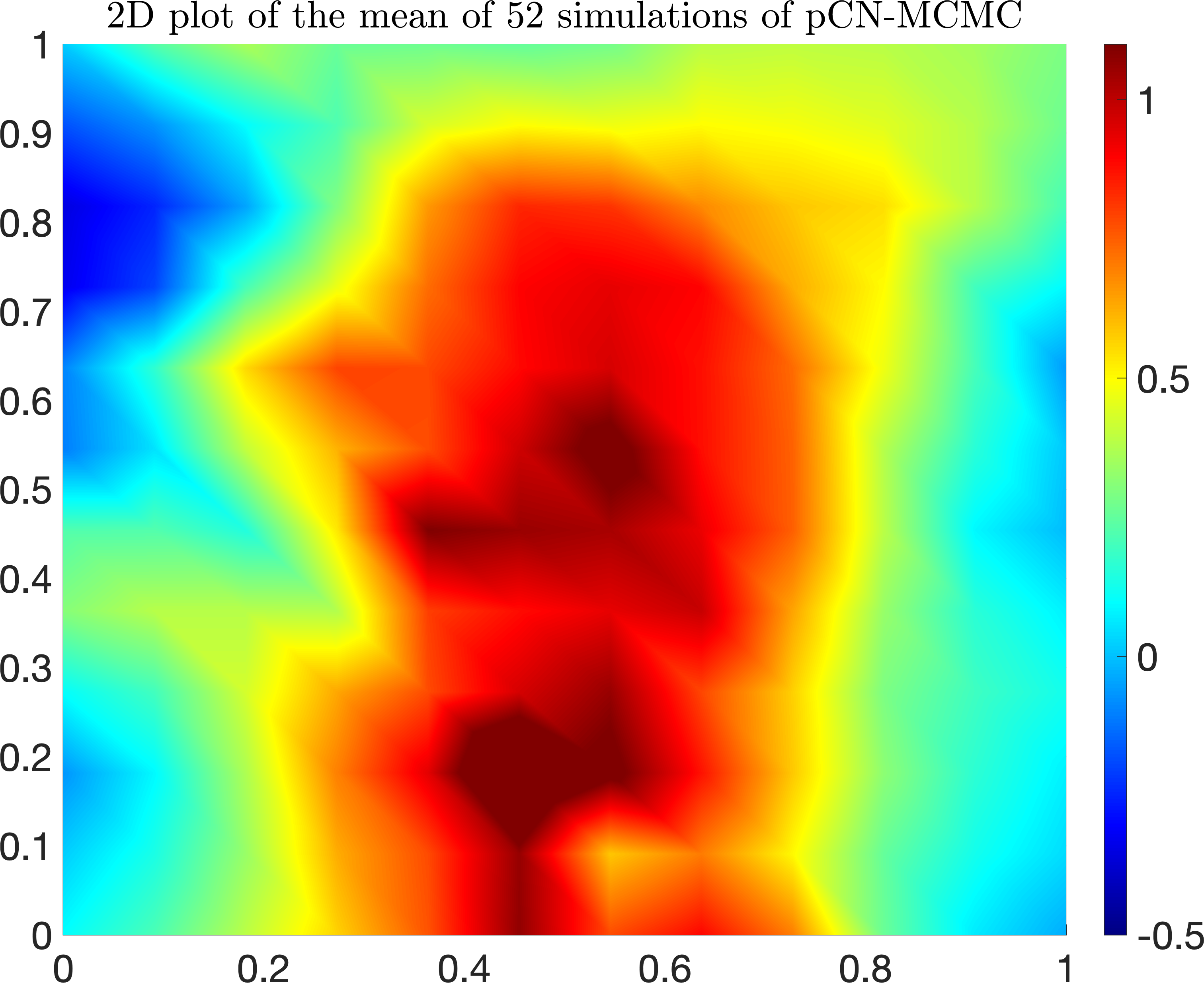}
}
\subcaptionbox{Single-level SFS results in 2D}{
\includegraphics[width =0.31\textwidth]{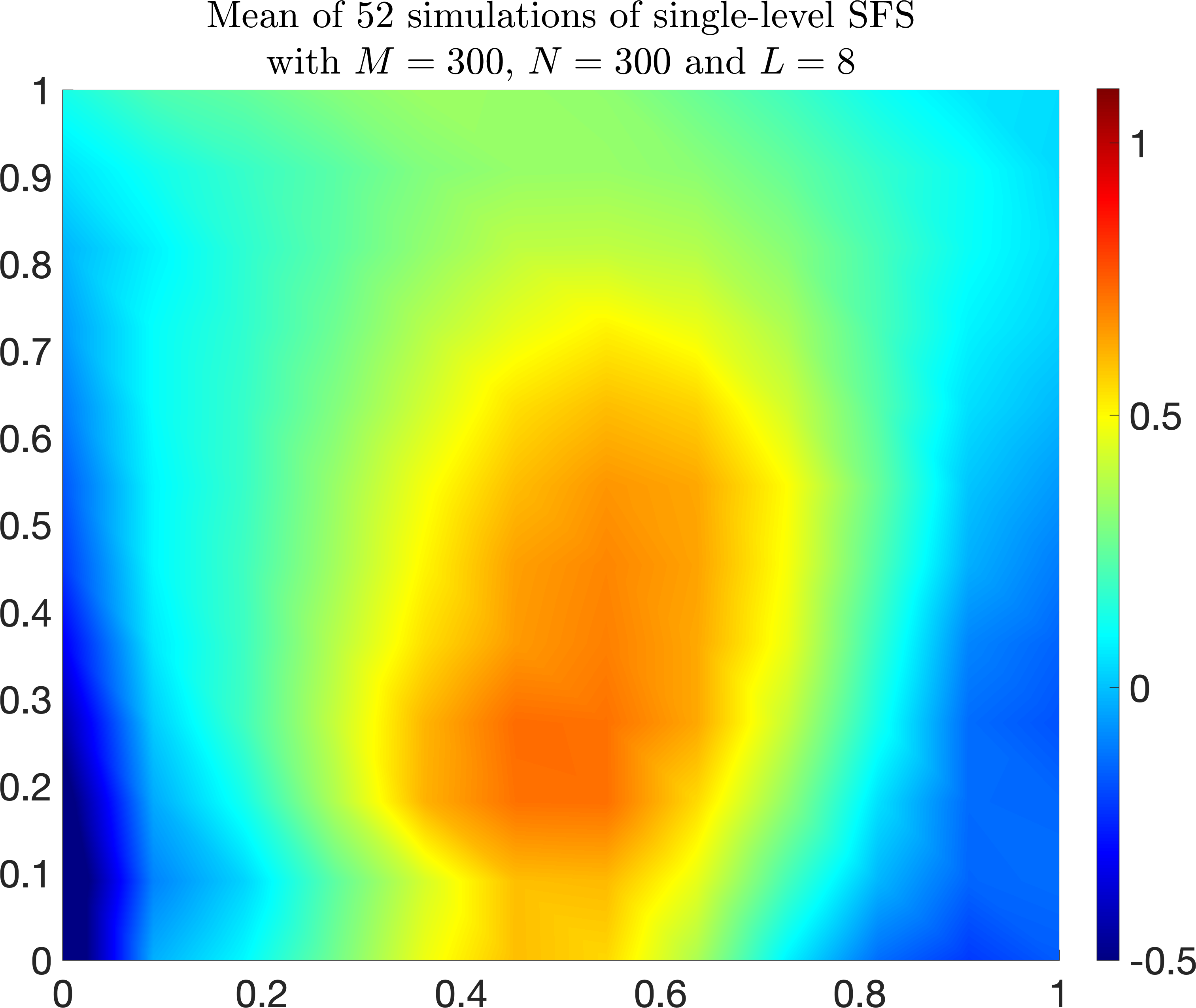}
}
\subcaptionbox{Unbiased SFS results in 2D}{
\includegraphics[width =0.31\textwidth]{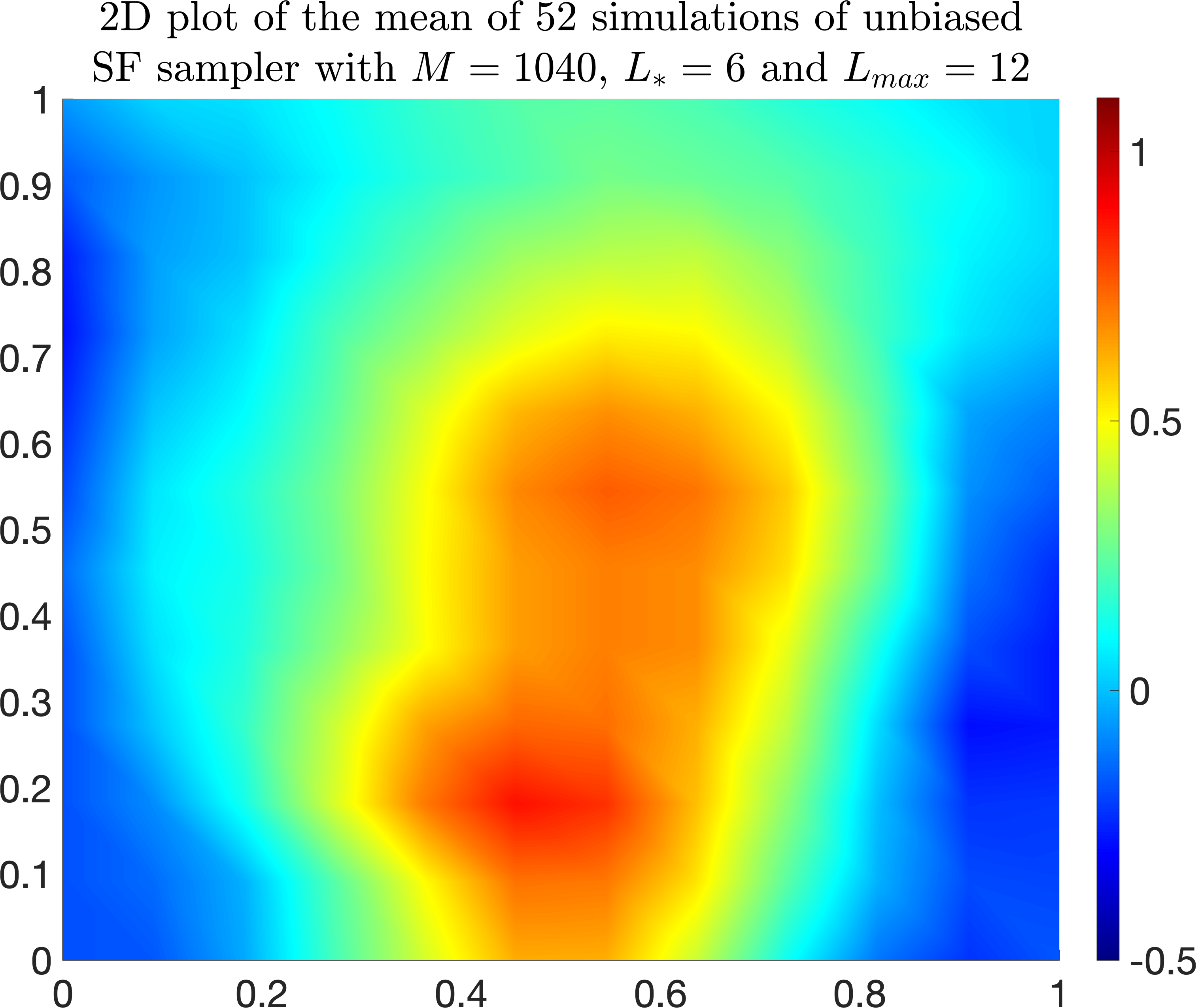}
}
\caption{Results of running the single-level SFS in \autoref{alg:basic_method} with $\hat{b}$ computed as in \eqref{eq:new_idea_drift1} and as in \autoref{alg:new_idea_ub}. It took around 33.5 hours to run \autoref{alg:new_idea_ub}, and for the same accuracy (RMSE for both is $\sim 3.9$), \autoref{alg:basic_method} took around 168 hours.}
\label{fig:inv_prob}
\end{figure}

Define a probability density for any fixed $(x,t)\in\mathbb{R}^d\times [0,1]$ as
\begin{align}
\label{eq:new_dinsity}
\pi_{x,t}(z) = \frac{f(x+\sqrt{1-t}z)\phi(z)}{\int_{\mathbb{R}^d}f(x+\sqrt{1-t}z)\phi(z)dz}.
\end{align}
Then we have that
$$
b(x,t) = \frac{1}{\sqrt{1-t}}\mathbb{E}_{\pi_{x,t}}[Z].
$$
Consider the discretized SDE in \eqref{eq:milstein} with $\hat{b}$ replaced by
\begin{align}
\label{eq:new_idea_drift1}
\hat{b}(\widetilde{X}_{k\Delta_l}^l,k\Delta_l) = \frac{1}{\sqrt{1-k\Delta_l}}\frac{1}{N}\sum_{i=1}^N Z^i,
\end{align}
where $\{Z_i\}_{i=1}^N$ are samples generated from $\pi_{\widetilde{X}_{k\Delta_l}^l,k\Delta_l}$ (by using any sampling method, e.g.  MCMC). Moreover, for any $(x,\tilde{x},t)\in\mathbb{R}^d\times\mathbb{R}^d\times[0,1)$ we can clearly write
\begin{align}
b(x,t) = \frac{1}{\sqrt{1-t}}\frac{\mathbb{E}_{\pi_{\tilde{x},t}}\Big[Z\frac{f(x+\sqrt{1-t}Z)}{f(\tilde{x}+\sqrt{1-t}Z)}\Big]}{\mathbb{E}_{\pi_{\tilde{x},t}}\Big[\frac{f(x+\sqrt{1-t}Z)}{f(\tilde{x}+\sqrt{1-t}Z)}\Big]}.
\label{eq:new_idea_IS}
\end{align}
This enables us to produce a coupled pair of (approximate) samples from a pair of Euler discretizations of \eqref{eq:main_diff} as explained in \autoref{alg:new_idea_coupling}. With the new identity in \eqref{eq:new_idea_IS}, we hope that computing the ratio $f(x+\sqrt{1-t}Z)/f(\tilde{x}+\sqrt{1-t}Z)$, $Z\sim \pi_{\tilde{x},t}$, on the log-scale will overcome the fixed point arithmetic issue discussed in the introduction of \autoref{sec:inv_prob}. Basically, we estimate the drift in \eqref{eq:milstein} using the identity in \eqref{eq:new_idea_drift1} at level zero and at the coarser level in the coupling step while using (\ref{eq:new_idea_weights}--\ref{eq:new_idea_fine_level}) at the finer level.

\subsubsection{Simulation Results}
We set the number of nodes in the mesh to $\hat{n} = 12^2$ and consider observations generated from the solution to the PDE on a finer mesh at 81 random nodes with $\sigma_y$ chosen such that a prescribed signal-to-noise ratio, defined as $\max\{p\}/\sigma_y$, is equal to 100. The true log-premeability field $u_{\text{truth}}$ used to generate the data is defined by the superposition of two Gaussians with covariances $0.05I_2$ and $0.07 I_2$, centered at $(0.5, 0.5)$ and $(0.5,0.95)$, with weights $\{0.6,0.2\}$ respectively.
The external force field $F$ is defined by a superposition of four weighted two-dimensional Gaussian bumps with covariance $0.03 I_2$, centered at (0.3, 0.5), (0.4, 0.3), (0.6, 0.9), and (0.7, 0.1), and with weights $\{2, -3, 1, -4\}$, respectively. In the covariance matrix of the prior, we take $\sigma = 1.5$ and $\alpha=0.9$. We ran 52 independent simulations of \autoref{alg:basic_method} with $\hat{b}$ computed as in \eqref{eq:new_idea_drift1} with $L = 8$, $N=300$ and $M=300$. We also ran 52 independent simulations of \autoref{alg:new_idea_ub} with $L_*=6$, $L_{\text{max}}=12$ and $M=1040$. The probability mass functions $\mathbb{P}_R(L)$ and $\mathbb{P}_{P|L}(P|L)$ are the same as in \autoref{subsubsec:simuls}. These parameters are chosen such that both algorithms give almost the same MSE. We remark that the reference log-permeability field was computed by taking the mean of 52 simulations of a preconditioned Crank-Nicolson (pCN) MCMC \cite{pCN} with $10^7$ samples a burn-in period of $10^3$. We also used the pCN-MCMC to sample from the density $\pi_{x,t}$ defined in \eqref{eq:new_dinsity}. Both algorithms were run on a workstation with 52 cores. Whilst aiming for similar precision, the computing cost of \autoref{alg:basic_method} was about a week, whereas for Algorithms \ref{alg:new_idea_coupling}-\ref{alg:new_idea_ub} the cost was a day and a half. The results are shown in \autoref{fig:inv_prob} and in this example we managed to obtain with Algorithms \ref{alg:new_idea_coupling}-\ref{alg:new_idea_ub} more accurate estimates at a fraction of the computational cost.

\subsubsection*{Acknowledgements}

AJ \& HR were supported by KAUST baseline funding.

\appendix

\section{Theoretical Results for the Proof of \autoref{theo:main_thm}}\label{app:theory}

The following Section contains a collection of Lemmata used to prove \autoref{theo:main_thm}. Recall that $C$ is a generic finite constant with value that could change upon each appearance, but will not depend upon $(l,N)$ and as $t$ is bounded by 1, not on $t$ either.  Recall, we use the notation that for a differentiable function $\psi:\mathbb{R}^d\rightarrow\mathbb{R}$, $\nabla_k \psi(x)=(\partial \psi/\partial x_k)(x)$, $k\in\{1,\dots,d\}$.

\begin{lem}\label{lem:stoch_lip}
Assume (A\ref{ass:2}). Then, there exists $C<\infty$ so that for any $(x,y,N,s)\in\mathbb{R}^{2d}\times\mathbb{N}\times[0,1]$ we have, almost surely, that
$$
\big\|\hat{b}(x,s)-\hat{b}(y,s)\big\|_2 \leq C\,\|x-y\|_2.
$$
\end{lem}

\begin{proof}
It suffices to prove that, almost surely, for each $j\in\{1,\dots,d\}$:
\begin{equation}\label{eq:lip_stoch}
\big|\big\{\hat{b}(x,s)-\hat{b}(y,s)\big\}_j\big| \leq C\,\|x-y\|_2.
\end{equation}
We have the simple decomposition
\begin{align*}
\frac{A(x)}{B(x)}-\frac{A(y)}{B(y)} = \tfrac{A(x)}{B(x)B(y)}\cdot \big(  B(y)-B(x) \big) + 
\tfrac{1}{B(y)}\cdot \big( A(x)- A(y)\big),
\end{align*}
where we have set
\begin{align*}
A(x) := \tfrac{1}{N}
\sum_{i=1}^N \nabla_j f(x+\sqrt{1-s}Z^i), 
\quad B(x) := \frac{1}{N}\sum_{i=1}^N f(x+\sqrt{1-s}Z^i).
\end{align*}
By (A\ref{ass:2}) a.~the terms
\begin{align*}
\frac{A(x)}{B(x)B(y)}, \quad \frac{1}{B(y)},
\end{align*}
are uniformly upper-bounded by a deterministic constant, so we have
\begin{align*}
\big|\big\{\hat{b}(x,s)-\hat{b}(y,s)\big\}_j\big|  \leq  C\, \big( \, |B(y)-B(x)|  + |A(y)-A(x)|\, \big).
\end{align*}
Applying the triangular inequality multiple times and (A\ref{ass:2}) b.~allows us to deduce the bound \eqref{eq:lip_stoch} and thus the proof is concluded.
\end{proof}

\begin{lem}
\label{lem:conv_mc}
Assume that (A\ref{ass:2}). Then, for any $p\in[1,\infty)$ there exists  a $C<\infty$ such that for any $(l,N,t)\in\mathbb{N}_0\times\mathbb{N}\times[0,1]$
$$
\mathbb{E}\,\big\|
\hat{b}(\widetilde{X}_{\tau_t^l}^{l},\tau_t^l)
-b(\widetilde{X}_{\tau_t^l}^{l},\tau_t^l)\big\|_2^p \leq \frac{C}{N^{p/2}}.
$$
\end{lem}

\begin{proof}
We consider a single co-ordinate of the vector $\hat{b}(\widetilde{X}_{\tau_t^l}^{l},\tau_t^l)
-b(\widetilde{X}_{\tau_t^l}^{l},\tau_t^l)$ and it suffices to bound
\begin{align}
\label{eq:suff}
\mathbb{E}\,\big|
\big\{\hat{b}(\widetilde{X}_{\tau_t^l}^{l},\tau_t^l)
-b(\widetilde{X}_{\tau_t^l}^{l},\tau_t^l)\big\}_j\big|^p
\end{align}
for any $j\in\{1,\dots,d\}$.
We have the simple decomposition
\begin{align*}
\frac{A^N}{B^N}-\frac{A}{B} = \tfrac{A^N}{B\cdot B^N}( B - B^{N} )  + 
 \tfrac{1}{B}( A^N - A ),
\end{align*}
where we have defined
\begin{gather*}
A^{N}:= \tfrac{1}{N}\sum_{i=1}^N \nabla_j f(\widetilde{X}_{\tau_t^l}^{l}+\sqrt{1-\tau_t^l}Z^i), \quad 
B^{N} :=  \tfrac{1}{N}\sum_{i=1}^N f(\widetilde{X}_{\tau_t^l}^{l}+\sqrt{1-\tau_t^l}Z^i), \\
A:= \mathbb{E}_{\phi}\,\big[\,\nabla_j f(\widetilde{X}_{\tau_t^l}^{l}+\sqrt{1-\tau_t^l}Z)\,], \quad 
B:= 
 \mathbb{E}_{\phi}\,\big[\,f(\widetilde{X}_{\tau_t^l}^{l}+\sqrt{1-\tau_t^l}Z)\,\big].
\end{gather*}
%
%
Thus, using (A\ref{ass:2}) and the $C_p$-inequality we have the following  upper-bound for \eqref{eq:suff}
\begin{align*}
 C\,\big(\,
\mathbb{E}\,| B
 - B^N|^p +
\mathbb{E}\,|A^N- A|^p\,
\big).
\end{align*}
As $\widetilde{X}_{\tau_t^l}^{l}$ is independent of $Z^1,\dots,Z^N$ we can use standard results for iid random variables to deduce that \eqref{eq:suff} is upper bounded by 
$C/N^{p/2}$
and the proof is now completed.
\end{proof}

\begin{lem}
\label{lem:conv_mc_marg}
Assume that (A\ref{ass:2}). Then, for any $p\in[1,\infty)$ there exists a $C<\infty$ such that for any $(l,N,t)\in\mathbb{N}_0\times\mathbb{N}\times[0,1]$
\begin{align}
\label{eq:term}
\mathbb{E}\,\big\|\widetilde{X}_{\tau_t^l}^{l,N}-\widetilde{X}_{\tau_t^l}^{l}\big\|_2^p \leq \frac{C}{N^{p/2}}.
\end{align}
\end{lem}

\begin{proof}
We have that
\begin{align*}
\widetilde{X}_{\tau_t^l}^{l,N}-\widetilde{X}_{\tau_t^l}^{l} = \int_{0}^{\tau_t^l}\big(\hat{b}(\widetilde{X}_{\tau_s^l}^{l,N},\tau_s^l)-
\hat{b}(\widetilde{X}_{\tau_s^l}^{l},\tau_s^l)\big)ds + 
\int_{0}^{\tau_t^l}\big(\hat{b}(\widetilde{X}_{\tau_s^l}^{l},\tau_s^l)-
b(\widetilde{X}_{\tau_s^l}^{l},\tau_s^l)\big)ds.
\end{align*}
Therefore, it easily follows that $\mathbb{E}\,\big\|\widetilde{X}_{\tau_t^l}^{l,N}-\widetilde{X}_{\tau_t^l}^{l}\big\|_2^p$ is upper bounded by 
\begin{align*}
  C\,
\int_{0}^{\tau_t^l}\Big\{\,
\mathbb{E}\,\big\|\hat{b}(\widetilde{X}_{\tau_s^l}^{l,N},\tau_s^l)-
\hat{b}(\widetilde{X}_{\tau_s^l}^{l},\tau_s^l)\big\|_2^p +  
\mathbb{E}\,\big\|
\hat{b}(\widetilde{X}_{\tau_s^l}^{l},\tau_s^l)
-b(\widetilde{X}_{\tau_s^l}^{l},\tau_s^l)\big\|_2^p\,\Big\}ds.
\end{align*}
For the first term in the integral one can use \autoref{lem:stoch_lip} and for the second one can use \autoref{lem:conv_mc}. Thus, one can deduce the following upper  bound for (\ref{eq:term})
$$
  C\,\Big( \int_0^t\mathbb{E}\,\big\|\widetilde{X}_{\tau_s^l}^{l,N}-\widetilde{X}_{\tau_s^l}^{l}\big\|_2^p \,ds+\frac{1}{N^{p/2}}\Big).
$$
So, the proof is concluded by applying  Gr\"onwall's inequality.
\end{proof}

\begin{lem}\label{lem:new_tech_lem}
Assume that (A\ref{ass:2}). Then, there exists a $C<+\infty$ such that for any $(l,N,s)\in\mathbb{N}_0\times\mathbb{N}\times[0,1]$ we have
\begin{align*}
\mathbb{E}\, \Big\|\big(\hat{b}(\widetilde{X}_{\tau_s^l}^{l,N},\tau_s^l)-
\hat{b}(\widetilde{X}_{\tau_s^{l-1}}^{l,N},\tau_s^{l-1})\big) - 
\big(\hat{b}(\widetilde{X}_{\tau_s^l}^{l},\tau_s^l)-\hat{b}(\widetilde{X}_{\tau_s^{l-1}}^{l},\tau_s^{l-1})\big)
\Big\|_2^2\leq \frac{C\Delta_l}{N}.
\end{align*}
\end{lem}

\begin{proof}
We will consider one co-ordinate of the vector
$$
T:= \big(\hat{b}(\widetilde{X}_{\tau_s^l}^{l,N},\tau_s^l)-\hat{b}(\widetilde{X}_{\tau_s^{l-1}}^{l,N},\tau_s^{l-1})\big) - 
\big(\hat{b}(\widetilde{X}_{\tau_s^l}^{l},\tau_s^l)-\hat{b}(\widetilde{X}_{\tau_s^{l-1}}^{l},\tau_s^{l-1})\big)
$$
as the argument is essentially the same across all co-ordinates.  We have that for any $j\in\{1,\dots,d\}$
\begin{align*}
T_j  = \Big(\frac{A^{l,N}}{B^{l,N}} - \frac{C^{l,N}}{D^{l,N}}\Big) - 
\Big(  \frac{A^l}{B^{l}} - \frac{C^{l}}{D^{l}}  \Big),
\end{align*}
where we have defined
\begin{gather*}
A^{l,N}:=  \tfrac{1}{N}\sum_{i=1}^N \nabla_j f(\widetilde{X}_{\tau_s^l}^{l,N}+\sqrt{1-\tau_s^l}Z^i); \quad 
B^{l,N}:= \tfrac{1}{N}\sum_{i=1}^Nf(\widetilde{X}_{\tau_s^l}^{l,N}+\sqrt{1-\tau_s^l}Z^i),\\
C^{l,N}:= \tfrac{1}{N}\sum_{i=1}^N \nabla_j f(\widetilde{X}_{\tau_s^{l-1}}^{l,N}+\sqrt{1-\tau_s^{l-1}}Z^i),\quad 
D^{l,N}:=  \tfrac{1}{N}\sum_{i=1}^N f(\widetilde{X}_{\tau_s^{l-1}}^{l,N}+\sqrt{1-\tau_s^{l-1}}Z^i),\\
A^l :=  \tfrac{1}{N}\sum_{i=1}^N \nabla_j f(\widetilde{X}_{\tau_s^l}^{l}+\sqrt{1-\tau_s^l}Z^i) ,\quad 
B^{l}:= \tfrac{1}{N}\sum_{i=1}^N f(\widetilde{X}_{\tau_s^l}^{l}+\sqrt{1-\tau_s^l}Z^i),\\
C^{l}:=  \tfrac{1}{N}\sum_{i=1}^N \nabla_j f(\widetilde{X}_{\tau_s^{l-1}}^{l}+\sqrt{1-\tau_s^{l-1}}Z^i),\quad
D^{l}:=\tfrac{1}{N}\sum_{i=1}^N f(\widetilde{X}_{\tau_s^{l-1}}^{l}+\sqrt{1-\tau_s^{l-1}}Z^i).
\end{gather*}
%
%
%
%
We  use \cite[Lemma C.5.]{mlpf} stating that for reals  $(a,b,c,d)$ and non-zero reals $(A,B,C,D)$
\begin{align}
\big(\tfrac{a}{A}&-\tfrac{b}{B}\big) - \big(\tfrac{c}{C}-\tfrac{d}{D}\big)  = \nonumber \\[0.2cm] &\tfrac{1}{A}\big((a-b)-(c-d)\big) - \tfrac{b}{AB}\big((A-B)-(C-D)\big) + \tfrac{1}{AC}(C-A)(c-d)  \nonumber\\[0.2cm]
& \qquad  -\tfrac{1}{AB}(b-d)(C-D) + \tfrac{d}{CBD}(B-D)(C-D) + \tfrac{d}{ACB}(A-C)(C-D). \label{eq:tech_lem1}
\end{align}
One can consider all six terms, individually, by using the $C_2$-inequality. However, as the terms
$$
\tfrac{1}{A}\big((a-b)-(c-d)\big), \quad \tfrac{b}{AB}\big((A-B)-(C-D)\big),
$$
are similar, we treat only the former. In addition, as the last four terms on the right side of \eqref{eq:tech_lem1} can be dealt with using similar calculations, we only deal with one of them. To that end, we seek to bound the two terms:
\begin{align*}
T_{j,1}  &:=  \mathbb{E}\,\Big[\,\Big( 
\tfrac{1}{(B^{l,N})^2}\cdot \big(   (A^{l,N} - C^{l,N}  ) - (  A^l - C^{l} )\big)^2\,\Big], \\
T_{j,2} &:=   \mathbb{E}\,\Big[\,     
\tfrac{1}{(B^{l})^2 (B^{l,N})^2}\cdot (B^{l} - B^{l,N})^2\cdot 
(    A^l - C^l )^2\,
\Big].
\end{align*}
%
%
%
%
For $T_{j,1}$ we easily obtain the upper-bound
\begin{align*}
 C\,\mathbb{E}\,\Big|\,
\big(
 \nabla_j f(Y_{\tau_s^l}^{l,N}) - 
 \nabla_j f(Y_{\tau_s^{l-1}}^{l,N})\big) -
\big(\nabla_j f(Y_{\tau_s^l}^{l})- 
 \nabla_j f(Y_{\tau_s^{l-1}}^{l})\big)
\Big|^2
\end{align*}
where we have set  
\begin{gather*}
Y_{\tau_s^l}^{l,N}:=\widetilde{X}_{\tau_s^l}^{l,N}+\sqrt{1-\tau_s^l}Z^1, \quad 
Y_{\tau_s^l}^{l}:=\widetilde{X}_{\tau_s^l}^{l}+\sqrt{1-\tau_s^l}Z^1, \\[0.2cm] 
Y_{\tau_s^{l-1}}^{l,N}:=\widetilde{X}_{\tau_s^{l-1}}^{l,N}+\sqrt{1-\tau_s^{l-1}}Z^1, \quad Y_{\tau_s^{l-1}}^{l}:=\widetilde{X}_{\tau_s^{l-1}}^{l}+\sqrt{1-\tau_s^{l-1}}Z^1.
\end{gather*}
%
%
Note now that 
\begin{align*}
\big(
 \nabla_j f(Y_{\tau_s^l}^{l,N}) - 
 \nabla_j f(Y_{\tau_s^{l-1}}^{l,N})\big) -
\big(\nabla_j f(Y_{\tau_s^l}^{l})- 
 \nabla_j f(Y_{\tau_s^{l-1}}^{l})\big)= \bar{T}_{j,1}(1) + \bar{T}_{j,1}(2),
\end{align*}
for the terms 
\begin{align*}
\bar{T}_{j,1}(1) & :=  \sum_{k=1}^d \int_0^1\nabla_k
\nabla_j f(Y_{\tau_s^{l-1}}^{l,N}+\lambda(Y_{\tau_s^{l}}^{l,N}-Y_{\tau_s^{l-1}}^{l,N}))\cdot 
\big\{(Y_{\tau_s^l}^{l,N}-Y_{\tau_s^{l-1}}^{l,N})- (Y_{\tau_s^l}^{l}-Y_{\tau_s^{l-1}}^{l})\big\}_k\,d\lambda, \\
\bar{T}_{j,1}(2) & :=  \sum_{k=1}^d\int_0^1
\Big(\nabla_k \nabla_j f(Y_{\tau_s^{l-1}}^{l,N}+\lambda(Y_{\tau_s^{l}}^{l,N}-Y_{\tau_s^{l-1}}^{l,N}))\\ &\qquad\qquad\qquad\qquad\qquad\qquad -\nabla_k \nabla_j f(Y_{\tau_s^{l-1}}^{l}+\lambda(Y_{\tau_s^{l}}^{l}-Y_{\tau_s^{l-1}}^{l})\Big)\cdot 
\{Y_{\tau_s^l}^{l}-Y_{\tau_s^{l-1}}^{l}\}_k\,d\lambda
\end{align*}
%
Then to obtain the required bound for $T_{j,1}$, it suffices to bound the second moments of $\bar{T}_{j,1}(1)$ and $\bar{T}_{j,1}(2)$ respectively.
For $\bar{T}_{j,1}(1)$  we have 
$$
\mathbb{E}\,\big|\,\bar{T}_{j,1}(1)\,\big|^2 \leq C\,\mathbb{E}\,\big\|(\widetilde{X}_{\tau_s^l}^{l,N}-\widetilde{X}_{\tau_s^{l-1}}^{l,N})-(\widetilde{X}_{\tau_s^l}^{l}-
\widetilde{X}_{\tau_s^{l-1}}^{l})\big\|_2^2.
$$
Also, since 
\begin{align}
\nonumber
(\widetilde{X}_{\tau_s^l}^{l,N}&-\widetilde{X}_{\tau_s^{l-1}}^{l,N})-(\widetilde{X}_{\tau_s^l}^{l}-\widetilde{X}_{\tau_s^{l-1}}^{l}) \\ &\qquad = \int_{\tau_s^{l-1}}^{\tau_s^{l}}
\Big(\big(\hat{b}(\widetilde{X}_{\tau_u^l}^{l,N},\tau_u^l)-\hat{b}(\widetilde{X}_{\tau_u^l}^{l},\tau_u^l)\big) + \big(\hat{b}(\widetilde{X}_{\tau_u^l}^{l},\tau_u^l)-b(\widetilde{X}_{\tau_u^l}^{l},\tau_u^l)\big)\Big)du
\label{eq:mc_inc_level}
\end{align}
one can follow similar arguments that were used to deduce \eqref{eq:nice_bound} to obtain
$$
\mathbb{E}\,\big|\,\bar{T}_{j,1}(1)\,\big|^2 \leq \frac{C\Delta_l^2}{N}.
$$
Using (A\ref{ass:2}) it follows that
\begin{align*}
\mathbb{E}\,\big|\,\bar{T}_{j,1}(2)\,\big|^2 \leq C\,\mathbb{E}\,\bigg[\,\Big(\big\|Y_{\tau_s^{l-1}}^{l,N}-Y_{\tau_s^{l-1}}^{l}\big\|_2^2
+\big\|(\widetilde{X}_{\tau_s^l}^{l,N}&-\widetilde{X}_{\tau_s^{l-1}}^{l,N})-(\widetilde{X}_{\tau_s^l}^{l}-\widetilde{X}_{\tau_s^{l-1}}^{l})\big\|_2^2\Big)
\\ &\qquad \qquad \qquad \times\big\|Y_{\tau_s^l}^{l}-Y_{\tau_s^{l-1}}^{l}\big\|_2^2\,\bigg].
\end{align*}
Using Cauchy-Schwarz, we obtain
\begin{align*}
\mathbb{E}\,\big|\,&\bar{T}_{j,1}(2)\,\big|^2 \\[0.2cm] &\leq  C\,\mathbb{E}\,\big[\,\|Y_{\tau_s^l}^{l}-Y_{\tau_s^{l-1}}^{l}\|_2^4\,\big]^{1/2}
\\
&\qquad \qquad\times \Big(\,\mathbb{E}\,\big[\,\|Y_{\tau_s^{l-1}}^{l,N}-Y_{\tau_s^{l-1}}^{l}\|_2^4\,\big]^{1/2} + 
\mathbb{E}\,\big[\,\big\|(\widetilde{X}_{\tau_s^l}^{l,N}-\widetilde{X}_{\tau_s^{l-1}}^{l,N})-(\widetilde{X}_{\tau_s^l}^{l}
-\widetilde{X}_{\tau_s^{l-1}}^{l})\big\|_2^4\,\big]^{1/2}\,\Big) \\[0.2cm] & = 
C\,\mathbb{E}\,\big[\,\big\|\widetilde{X}_{\tau_s^l}^{l}-\widetilde{X}_{\tau_s^{l-1}}^{l}+\{\sqrt{1-\tau_s^l}-\sqrt{1-\tau_s^{l-1}}\}Z^1\big\|^4\,\big]^{1/2}
\\ &\qquad \qquad\times \Big(\,\mathbb{E}\,\big[\,\big\|\widetilde{X}_{\tau_s^{l-1}}^{l,N}-\widetilde{X}_{\tau_s^{l-1}}^{l}\big\|_2^4\,\big]^{1/2} + 
\mathbb{E}\,\big[\,\big\|(\widetilde{X}_{\tau_s^l}^{l,N}-\widetilde{X}_{\tau_s^{l-1}}^{l,N})-(\widetilde{X}_{\tau_s^l}^{l}-
\widetilde{X}_{\tau_s^{l-1}}^{l})\big\|_2^4\,\big]^{1/2}\,\Big).
\end{align*}
For the first factor term in the above upper bound, using standard results on Euler discretizations and Gaussian random variables, we have
$$
\mathbb{E}\,\Big[\,\big\|(\widetilde{X}_{\tau_s^l}^{l}-\widetilde{X}_{\tau_s^{l-1}}^{l})+(\sqrt{1-\tau_s^l}-\sqrt{1-\tau_s^{l-1}})Z^1\big\|_2^4\,\Big]^{1/2} \leq C\Delta_l.
$$
Then using \autoref{lem:conv_mc_marg} and the above arguments one obtains
$$
\mathbb{E}\,\big[\,\bar{T}_{j,1}(2)^2\,\big] \leq \frac{C\Delta_l}{N}.
$$
Thus we can conclude that 
$$
T_{j,1} \leq \frac{C\Delta_l}{N}.
$$
For $T_{j,2}$, using (A\ref{ass:2}) and Cauchy-Schwarz it follows that
\begin{align*}
&T_{j,2} \leq C\,\mathbb{E}\,\Big[\,\Big(f(\widetilde{X}_{\tau_s^l}^{l}+\sqrt{1-\tau_s^l}Z^1)-f(\widetilde{X}_{\tau_s^l}^{l,N}+\sqrt{1-\tau_s^l}Z^1)\Big)^4\,\Big]^{1/2}
\\ &\qquad \qquad \qquad \times 
\mathbb{E}\,\Big[\,\Big(\nabla_j f(\widetilde{X}_{\tau_s^l}^{l}+\sqrt{1-\tau_s^l}Z^1)-\nabla_j f(\widetilde{X}_{\tau_s^{l-1}}^{l}+\sqrt{1-\tau_s^{l-1}}Z^1)\Big)^4\,
\Big]^{1/2}.
\end{align*}
For the first factor one can use (A\ref{ass:2}) b.~and \autoref{lem:conv_mc_marg}, and for the second  one can use  (A\ref{ass:2})~b.~and standard properties of Euler-discretizations to give
$$
T_{j,2} \leq \frac{C\Delta_l}{N}.
$$
This completes the proof.
\end{proof}

\begin{lem}\label{lem:new_tech_lem1}
Assume that (A\ref{ass:2}). Then, there exists a $C<+\infty$ such that for any $(l,N,s)\in\mathbb{N}^2\times[0,1]$ we have
\begin{align*}
&\mathbb{E}\,\Big\|\big(\hat{b}(\widetilde{X}_{\tau_s^{l-1}}^{l,N},\tau_s^{l-1})-\hat{b}(\widetilde{X}_{\tau_s^{l-1}}^{l-1,N},\tau_s^{l-1})\big) - 
\big(\hat{b}(\widetilde{X}_{\tau_s^{l-1}}^{l},\tau_s^{l-1})-\hat{b}(\widetilde{X}_{\tau_s^{l-1}}^{l-1},\tau_s^{l-1})\big)
\Big\|_2^2\leq  \\[0.2cm]
&\qquad\qquad\qquad\qquad\qquad\qquad \qquad C\,\Big(\,\frac{\Delta_l^2}{N}+
\mathbb{E}\,\Big\|\big(\widetilde{X}_{\tau_s^l}^{l,N}-\widetilde{X}_{\tau_s^{l-1}}^{l-1,N}\big) - 
\big(\widetilde{X}_{\tau_s^l}^{l}-\widetilde{X}_{\tau_s^{l-1}}^{l-1}\big)
\Big\|_2^2\,\Big).
\end{align*}
\end{lem}

\begin{proof}
The proof is much the same as that of \autoref{lem:new_tech_lem}. The only difference is that the term 
$
\Delta_l^2/N
$
occurs as one considers the processes at the same time instance and that one will obtain an additive term in the upper-bound of the type:
$$
\mathbb{E}\,\Big\|\big(\widetilde{X}_{\tau_s^{l-1}}^{l,N}-\widetilde{X}_{\tau_s^{l-1}}^{l-1,N}\big) - 
\big(\widetilde{X}_{\tau_s^{l-1}}^{l}-\widetilde{X}_{\tau_s^{l-1}}^{l-1}\big)
\Big\|_2^2.
$$
This latter term is upper-bounded by
$$
C\,\Big(\,
\mathbb{E}\,\Big\|\big(\widetilde{X}_{\tau_s^l}^{l,N}-\widetilde{X}_{\tau_s^{l-1}}^{l,N}\big)-\big(\widetilde{X}_{\tau_s^l}^{l}-\widetilde{X}_{\tau_s^{l-1}}^{l}\big)\Big\|_2^2
+
\mathbb{E}\,\Big\|\big(\widetilde{X}_{\tau_s^l}^{l,N}-\widetilde{X}_{\tau_s^{l-1}}^{l-1,N}\big) - 
\big(\widetilde{X}_{\tau_s^l}^{l}-\widetilde{X}_{\tau_s^{l-1}}^{l-1}\big)
\Big\|_2^2
\,\Big).
$$
The first-term  above is easily proved to be $\mathcal{O}(\Delta_l^2/N)$ (see \eqref{eq:mc_inc_level} and the subsequent argument) and this concludes the proof.
\end{proof}

\begin{lem}\label{lem:tech_lem}
Assume (A\ref{ass:2}). Then there exists a $C<+\infty$ such that for any $(l,N,t)\in\mathbb{N}^2\times[0,1]$ we have
\begin{align*}
&\mathbb{E}\,\Big\|
\int_0^{t}\Big(\big(\hat{b}(\widetilde{X}_{\tau_s^l}^{l,N},\tau_s^l)-\hat{b}(\widetilde{X}_{\tau_s^{l-1}}^{l-1,N},\tau_s^{l-1})\big) - 
\big(\hat{b}(\widetilde{X}_{\tau_s^l}^{l},\tau_s^l)-\hat{b}(\widetilde{X}_{\tau_s^{l-1}}^{l-1},\tau_s^{l-1})\big)\Big)ds
\Big\|_2^2\leq 
\\[0.2cm]
&\qquad\qquad\qquad\qquad\qquad\qquad C\,\Big(\,\frac{\Delta_l}{N}
+ \int_0^t
\mathbb{E}\,\Big\|\big(\widetilde{X}_{\tau_s^l}^{l,N}-\widetilde{X}_{\tau_s^{l-1}}^{l-1,N}\big) - 
\big(\widetilde{X}_{\tau_s^l}^{l}-\widetilde{X}_{\tau_s^{l-1}}^{l-1}\big)
\Big\|_2^2 ds\,\Big).
\end{align*}
\end{lem}

\begin{proof}
It is simple to establish that
\begin{align*}
&\mathbb{E}\,\Big\|
\int_0^{t}\Big(\big(\hat{b}(\widetilde{X}_{\tau_s^l}^{l,N},\tau_s^l)-\hat{b}(\widetilde{X}_{\tau_s^{l-1}}^{l-1,N},\tau_s^{l-1})\big)- 
\big(\hat{b}(\widetilde{X}_{\tau_s^l}^{l},\tau_s^l)-\hat{b}(\widetilde{X}_{\tau_s^{l-1}}^{l-1},\tau_s^{l-1})\big)\Big)ds
\Big\|_2^2\leq 
\\[0.2cm]
&
\qquad\qquad  C\,\int_0^t\mathbb{E}\,\Big\|
\big(\hat{b}(\widetilde{X}_{\tau_s^l}^{l,N},\tau_s^l)-\hat{b}(\widetilde{X}_{\tau_s^{l-1}}^{l-1,N},\tau_s^{l-1})\big) - 
\big(\hat{b}(\widetilde{X}_{\tau_s^l}^{l},\tau_s^l)-\hat{b}(\widetilde{X}_{\tau_s^{l-1}}^{l-1},\tau_s^{l-1})\big)
\Big\|_2^2ds
\end{align*}
so we focus on the term inside the integrand. 
We have
\begin{align*}
&\mathbb{E}\,\Big\|
\big(\hat{b}(\widetilde{X}_{\tau_s^l}^{l,N},\tau_s^l)-\hat{b}(\widetilde{X}_{\tau_s^{l-1}}^{l-1,N},\tau_s^{l-1})\big) - 
\big(\hat{b}(\widetilde{X}_{\tau_s^l}^{l},\tau_s^l)-\hat{b}(\widetilde{X}_{\tau_s^{l-1}}^{l-1},\tau_s^{l-1})\big)
\Big\|_2^2 \leq
\\[0.2cm] &
\qquad \qquad C\,\bigg(\,
\mathbb{E}\,\Big\|\big(\hat{b}(\widetilde{X}_{\tau_s^l}^{l,N},\tau_s^l)-\hat{b}(\widetilde{X}_{\tau_s^{l-1}}^{l,N},\tau_s^{l-1})\big) - 
\big(\hat{b}(\widetilde{X}_{\tau_s^l}^{l},\tau_s^l)-\hat{b}(\widetilde{X}_{\tau_s^{l-1}}^{l},\tau_s^{l-1})\big)
\Big\|_2^2 + \\
&
\qquad \qquad\qquad\qquad \mathbb{E}\,\Big\|\big(\hat{b}(\widetilde{X}_{\tau_s^{l-1}}^{l,N},\tau_s^{l-1})-\hat{b}(\widetilde{X}_{\tau_s^{l-1}}^{l-1,N},\tau_s^{l-1})\big) - 
\big(\hat{b}(\widetilde{X}_{\tau_s^{l-1}}^{l},\tau_s^{l-1})-\hat{b}(\widetilde{X}_{\tau_s^{l-1}}^{l-1},\tau_s^{l-1})\big)
\Big\|_2^2\,\bigg).
\end{align*}
So, the proof is concluded by applying Lemmata \ref{lem:new_tech_lem}-\ref{lem:new_tech_lem1}.
\end{proof}

\begin{lem}\label{lem:tech_lem3}
Assume (A\ref{ass:2}). Then there exists a $C<+\infty$ such that for any $(l,N,s)\in\mathbb{N}^2\times[0,1]$ we have
$$
\mathbb{E}\,\Big\|
\big(\hat{b}(\widetilde{X}_{\tau_s^l}^{l},\tau_s^l)-\hat{b}(\widetilde{X}_{\tau_s^{l-1}}^{l},\tau_s^{l-1})\big) - 
\big(b(\widetilde{X}_{\tau_s^{l}}^{l},\tau_s^{l})-b(\widetilde{X}_{\tau_s^{l-1}}^{l},\tau_s^{l-1})\big)
\Big\|_2^2\ \leq \frac{C\Delta_l}{N}.
$$
\end{lem}

\begin{proof}
This can be proved using the same identity (equation \eqref{eq:tech_lem1}) as \autoref{lem:new_tech_lem}. The subsequent calculations are much simpler than the proof of \autoref{lem:new_tech_lem} and are hence omitted.
\end{proof}

\begin{lem}\label{lem:tech_lem4}
Assume (A\ref{ass:2}). Then there exists a $C<+\infty$ such that for any $(l,N,t)\in\mathbb{N}^2\times[0,1]$ we have
$$
\mathbb{E}\,\Big\|
\big(\hat{b}(\widetilde{X}_{\tau_s^{l-1}}^{l},\tau_s^{l-1})-\hat{b}(\widetilde{X}_{\tau_s^{l-1}}^{l-1},\tau_s^{l-1})\big) - 
\big(b(\widetilde{X}_{\tau_s^{l-1}}^{l},\tau_s^{l-1})-b(\widetilde{X}_{\tau_s^{l-1}}^{l-1},\tau_s^{l-1})\big)
\Big\|_2^2 \leq \frac{C\Delta_l^2}{N}.
$$
\end{lem}

\begin{proof}
We will consider one co-ordinate of the vector
$$
T:= \big(\hat{b}(\widetilde{X}_{\tau_s^{l-1}}^{l},\tau_s^{l-1})-\hat{b}(\widetilde{X}_{\tau_s^{l-1}}^{l-1},\tau_s^{l-1})\big) - 
\big(b(\widetilde{X}_{\tau_s^{l-1}}^{l},\tau_s^{l-1})-b(\widetilde{X}_{\tau_s^{l-1}}^{l-1},\tau_s^{l-1})\big),
$$
as the argument is essentially the same across all co-ordinates.  We have that for any $j\in\{1,\dots,d\}$
\begin{align*}
T_j  = \Big(\frac{C^{l}}{D^{l}} - \frac{E^{l}}{F^{l}}\Big) - 
\Big(  \frac{C}{D} - \frac{E}{F}  \Big),
\end{align*}
where we have defined
\begin{gather*}
C^{l}:=  \tfrac{1}{N}\sum_{i=1}^N \nabla_j f(\widetilde{X}_{\tau_s^{l-1}}^{l}+\sqrt{1-\tau_s^{l-1}}Z^i); \quad 
D^{l}:= \tfrac{1}{N}\sum_{i=1}^Nf(\widetilde{X}_{\tau_s^{l-1}}^{l}+
\sqrt{1-\tau_s^{l-1}}Z^i),\\
E^{l}:=  \tfrac{1}{N}\sum_{i=1}^N \nabla_j f(\widetilde{X}_{\tau_s^{l-1}}^{l-1}+\sqrt{1-\tau_s^{l-1}}Z^i); \quad 
F^{l}:= \tfrac{1}{N}\sum_{i=1}^Nf(\widetilde{X}_{\tau_s^{l-1}}^{l-1}+
\sqrt{1-\tau_s^{l-1}}Z^i),\\[0.2cm]
C:= \mathbb{E}_\phi\,\big[\,\nabla_j f(\widetilde{X}_{\tau_s^{l-1}}^{l}+\sqrt{1-\tau_s^{l-1}}Z)\,\big] ,\quad 
D:=  \mathbb{E}_\phi\,\big[\, f(\widetilde{X}_{\tau_s^{l-1}}^{l}+\sqrt{1-\tau_s^{l-1}}Z)\,\big],\\[0.4cm]
E:= \mathbb{E}_\phi\,\big[\,\nabla_j f(\widetilde{X}_{\tau_s^{l-1}}^{l-1}+\sqrt{1-\tau_s^{l-1}}Z)\,\big],\quad
F:=\mathbb{E}_\phi\,\big[\, f(\widetilde{X}_{\tau_s^{l-1}}^{l-1}+\sqrt{1-\tau_s^{l-1}}Z)\,\big].
\end{gather*}
Again, we use \eqref{eq:tech_lem1} and just give a proof for two terms
\begin{align*}
T_{j,1} & := \mathbb{E}\,\Big[\,
\frac{1}{(D^{l})^2}
\big((C^l-E^l ) -  (C-E)\big)^2 \,\Big], \\ 
T_{j,2} & := \mathbb{E}\,\Big[\,
\frac{1}{(D^l)^2 D^2}(
D-D^l)^2
(C-E)^2
\,\Big].
\end{align*}
For $T_{j,1}$ applying  (A\ref{ass:2}) a.~and using the fact that $(\widetilde{X}_{\tau_s^{l-1}}^{l},\widetilde{X}_{\tau_s^{l-1}}^{l-1})$ are independent of the iid~$Z^{1},\dots,Z^N$, we have
$$
T_{j,1} \leq \frac{C}{N}\,\mathbb{E}\,\Big|\,\nabla_j f(\widetilde{X}_{\tau_s^{l-1}}^{l}+\sqrt{1-\tau_s^{l-1}}Z^1) -\nabla_j f(\widetilde{X}_{\tau_s^{l-1}}^{l-1}+\sqrt{1-\tau_s^{l-1}}Z^1)\,\Big|^2.
$$
Then using (A\ref{ass:2}) b.~along with standard results for strong errors of diffusions we have
$$
T_{j,1} \leq \frac{C\Delta_l^2}{N}.
$$
For $T_{j,2}$ applying  (A\ref{ass:2}) a.~and the Cauchy-Schwarz inequality 
\begin{align*}
T_{j,2} & \leq  C\,\mathbb{E}\Big[\,\Big(\,\mathbb{E}_\phi\,\big[\,
f(\widetilde{X}_{\tau_s^{l-1}}^{l}+\sqrt{1-\tau_s^{l-1}}Z)\,\big]-
\tfrac{1}{N}\sum_{i=1}^N f(\widetilde{X}_{\tau_s^{l-1}}^{l}+\sqrt{1-\tau_s^{l-1}}Z^i)\,\Big)^4\,\Big]^{1/2} \\
& \qquad \qquad \qquad \times\mathbb{E}\,\Big[\,\Big(\,\mathbb{E}_\phi\,\big[\,\nabla_j f(\widetilde{X}_{\tau_s^{l-1}}^{l}+\sqrt{1-\tau_s^{l-1}}Z)-\nabla_j f(\widetilde{X}_{\tau_s^{l-1}}^{l-1}+\sqrt{1-\tau_s^{l-1}}Z)\,\big]\,\Big)^4\,\Big]^{1/2}.
\end{align*}
For the first term on the right-hand-side one can use standard results for iid  random variables and for the second  (A\ref{ass:2}) a.~along with standard results for strong errors of diffusions
to yield
$$
T_{j,2} \leq \frac{C\Delta_l^2}{N}.
$$
From here, one can complete the proof fairly easily.
\end{proof}

\begin{lem}\label{lem:tech_lem2}
Assume (A\ref{ass:2}). Then there exists a $C<+\infty$ such that for any $(l,N,t)\in\mathbb{N}^2\times[0,1]$ we have
$$
\mathbb{E}\,\Big\|
\int_0^{t}\Big(
\big(\hat{b}(\widetilde{X}_{\tau_s^l}^{l},\tau_s^l)-\hat{b}(\widetilde{X}_{\tau_s^{l-1}}^{l-1},\tau_s^{l-1})\big) - 
\big(b(\widetilde{X}_{\tau_s^{l}}^{l},\tau_s^{l})-b(\widetilde{X}_{\tau_s^{l-1}}^{l-1},\tau_s^{l-1})\big)\Big)ds
\Big\|_2^2\leq \frac{C\Delta_l}{N}.
$$
\end{lem}

\begin{proof}
The proof is essentially the same as that for \autoref{lem:tech_lem}, except one must use Lemmata \ref{lem:tech_lem3}-\ref{lem:tech_lem4} instead of Lemmata \ref{lem:new_tech_lem}-\ref{lem:new_tech_lem1}; therefore the proof is omitted.
\end{proof}

\end{document}